\documentclass[11pt]{article}%
\usepackage{amsmath}
\usepackage{amssymb}
\usepackage{amsfonts}

\usepackage{cite}
\usepackage{graphicx}
\usepackage{float}
\usepackage{longtable}
\usepackage[utf8]{inputenc}
\usepackage{hyperref}

\usepackage{algorithmicx}
\usepackage[ruled,vlined]{algorithm2e}
\SetAlgoCaptionSeparator{ }

\usepackage{fullpage}

\newcommand{\fref}[1]{Fig.~\ref{#1}}

\newcommand{\tref}[1]{Table~\ref{#1}}
\newcommand{\sref}[1]{Section~\ref{#1}}

\newenvironment{algo}[1][!htbp]
  {
   \begin{algorithm}[#1]%
  }{\end{algorithm}}

\providecommand{\U}[1]{\protect\rule{.1in}{.1in}}
\newtheorem{theorem}{Theorem}

\newtheorem{definition}{Definition}

\newtheorem{lemma}{Lemma}

\newtheorem{proposition}{Proposition}

\newenvironment{proof}[1][Proof]{\textbf{#1.} }{\ \rule{0.5em}{0.5em}}

\begin{document}

\title{Training Optimization for Gate-Model Quantum Neural Networks}
\author{Laszlo Gyongyosi\thanks{School of Electronics and Computer Science, University of Southampton, Southampton SO17 1BJ, U.K., and Department of Networked Systems and Services, Budapest University of Technology and Economics, 1117 Budapest, Hungary, and MTA-BME Information Systems Research Group, Hungarian Academy of Sciences, 1051 Budapest, Hungary.}
\and Sandor Imre\thanks{Department of Networked Systems and Services, Budapest University of Technology and Economics, 1117 Budapest, Hungary.}}

\date{}

\maketitle
\begin{abstract}
Gate-based quantum computations represent an essential to realize near-term quantum computer architectures. A gate-model quantum neural network (QNN) is a QNN implemented on a gate-model quantum computer, realized via a set of unitaries with associated gate parameters. Here, we define a training optimization procedure for gate-model QNNs. By deriving the environmental attributes of the gate-model quantum network, we prove the constraint-based learning models. We show that the optimal learning procedures are different if side information is available in different directions, and if side information is accessible about the previous running sequences of the gate-model QNN. The results are particularly convenient for gate-model quantum computer implementations.
\end{abstract}

\section{Introduction}
\label{sec1}
Gate-based quantum computers represent an implementable way to realize experimental quantum computations on near-term quantum computer architectures \cite{ref1,ref2,ref3,ref4,ref5,ref6,ref7,ref8,pres,harr,a1,refibm,br}. In a gate-model quantum computer, the transformations are realized by quantum gates, such that each quantum gate is represented by a unitary operation \cite{ref9,ref10,ref11,ref12,ref13,ref14,ref15,ref16,ref18,ref20,ref22,ref37,ref38}. An input quantum state is evolved through a sequence of unitary gates and the output state is then assessed by a measurement operator \cite{ref9,ref10,ref11,ref12}. Focusing on gate-model quantum computer architectures is motivated by the successful demonstration of the practical implementations of gate-model quantum computers \cite{ref4,ref5,ref6,ref7,ref8}, and several important developments for near-term gate-model quantum computations are currently in progress. Another important aspect is the application of gate-model quantum computations in the near-term quantum devices of the quantum Internet \cite{ref48,ref49,ref50,addp1,addp2,ref19,ref17,ref21,ref41,ref42,ref43,ref44,ref45,ref46,ref47,ref51,ref52}.   

A quantum neural network (QNN) is formulated by a set of quantum operations and connections between the operations with a particular weight parameter \cite{ref9,ref29,ref30, ref35,ref36,ref37,ref38}. Gate-model QNNs refer to QNNs implemented on gate-model quantum computers \cite{ref9}. As a corollary, gate-model QNNs have a crucial experimental importance since these network structures are realizable on near-term quantum computer architectures. The core of a gate-model QNN is a sequence of unitary operations. A gate-model QNN consists of a set of unitary operations and communication links that are used for the propagation of quantum and classical side information in the network for the related calculations of the learning procedure. The unitary transformations represent quantum gates parameterized by a variable referred to as gate parameter (weight). The inputs of the gate-model QNN structure are a computational basis state and an auxiliary quantum system that serves a readout state in the output measurement phase. Each input state is associated with a particular label. In the modeled learning problem, the training of the gate-model QNN aims to learn the values of the gate parameters associated with the unitaries so that the predicted label is close to a true label value of the input (i.e., the difference between the predicted and true values is minimal). This problem, therefore, formulates an objective function that is subject to minimization. In this setting, the training of the gate-model QNN aims to learn the label of a general quantum state. 

In artificial intelligence, machine learning \cite{ref1,ref2,ref3,ref15,ref20,ref23,ref28,ref29,ref30,ref31,ref32,ref33,ref34} utilizes statistical methods with measured data to achieve a desired value of an objective function associated with a particular problem. A learning machine is an abstract computational model for the learning procedures. A constraint machine is a learning machine that works with constraint, such that the constraints are characterized and defined by the actual environment \cite{ref23}. 

The proposed model of a gate-model quantum neural network assumes that quantum information can only be propagated forward direction from the input to the output, and classical side information is available via classical links. The classical side information is processed further via a post-processing unit after the measurement of the output. In the general gate-model QNN scenario, it is assumed that classical side information can be propagated arbitrarily in the network structure, and there is no available side information about the previous running sequences of the gate-model QNN structure. The situation changes, if side information propagates only backward direction and side information about the previous running sequences of the network is also available. The resulting network model is called gate-model recurrent quantum neural network (RQNN). 

Here, we define a constraint-based training optimization method for gate-model QNNs and RQNNs, and propose the computational models from the attributes of the gate-model quantum network environment. We show that these structural distinctions lead to significantly different computational models and learning optimization. By using the constraint-based computational models of the QNNs, we prove the optimal learning methods for each network---nonrecurrent and recurrent gate-model QNNs---vary. Finally, we characterize optimal learning procedures for each variant of gate-model QNNs.

The novel contributions of our manuscript are as follows.
\begin{itemize}
\item \textit{We study the computational models of nonrecurrent and recurrent gate-model QNNs realized via an arbitrary number of unitaries.}
\item \textit{We define learning methods for nonrecurrent and recurrent gate-model QNNs.}
\item \textit{We prove the optimal learning for nonrecurrent and recurrent gate-model QNNs.}
\end{itemize}
This paper is organized as follows. In \sref{relw}, the related works are summarized. \sref{sec2} defines the system model and the parameterization of the learning optimization problem. \sref{sec4} proves the computational models of gate-model QNNs. \sref{sec5} provides learning optimization results. Finally, \sref{sec6} concludes the paper. Supplemental information is included in the Appendix. 
\section{Related Works}
\label{relw}
\subsection{Gate-Model Quantum Computers}
A theoretical background on the realizations of quantum computations in a gate-model quantum computer environment can be found in \cite{ref10} and \cite{ref11}. For a summary on the related references \cite{ref10,ref11,su,ref12,br,refa6,pres,harr,a1}, we suggest \cite{dense}.

\subsection{Quantum Neural Networks}
In \cite{ref9}, the formalism of a gate-model quantum neural network is defined. The gate-model quantum neural network is a quantum neural network implemented on gate-model quantum computer. A particular problem analyzed by the authors is the classification of classical data sets which consist of bitstrings with binary labels.

In \cite{ref36}, the authors studied the subject of quantum deep learning. As the authors found, the application of quantum computing can reduce the time required to train a deep restricted Boltzmann machine. The work also concluded that quantum computing provides a strong framework for deep learning, and the application of quantum computing can lead to significant performance improvements in comparison to classical computing.

In \cite{ref29}, the authors defined a quantum generalization of feedforward neural networks. In the proposed system model, the classical neurons are generalized to being quantum reversible. As the authors showed, the defined quantum network can be trained efficiently using gradient descent to perform quantum generalizations of classical tasks. 

In \cite{ref30}, the authors defined a model of a quantum neuron to perform machine learning tasks on quantum computers. The authors proposed a small quantum circuit to simulate neurons with threshold activation. As the authors found, the proposed quantum circuit realizes a “quantum neuron”. The authors showed an application of the defined quantum neuron model in feedforward networks. The work concluded that the quantum neuron model can learn a function if trained with superposition of inputs and the corresponding output. The proposed training method also suffices to learn the function on all individual inputs separately. 

In \cite{ref37}, the authors studied the structure of artificial quantum neural network. The work focused on the model of quantum neurons and studied the logical elements and tests of convolutional networks. The authors defined a model of an artificial neural network that uses quantum-mechanical particles as a neuron, and set a Monte-Carlo integration method to simulate the proposed quantum-mechanical system. The work also studied the implementation of logical elements based on introduced quantum particles, and the implementation of a simple convolutional network.

In \cite{ref38}, the authors defined the model of a universal quantum perceptron as efficient unitary approximators. The authors studied the implementation of a quantum perceptron with a sigmoid activation function as a reversible many-body unitary operation. In the proposed system model, the response of the quantum perceptron is parameterized by the potential exerted by other neurons. The authors showed that the proposed quantum neural network model is a universal approximator of continuous functions, with at least the same power as classical neural networks.

\subsection{Quantum Machine Learning}
In \cite{sat}, the authors analyzed a Markov process connected to a classical probabilistic algorithm \cite{sch}. A performance evaluation also has been included in the work to compare the performance of the quantum and classical algorithm.
 
In \cite{ref15}, the authors studied quantum algorithms for supervised and unsupervised machine learning. This particular work focuses on the problem of cluster assignment and cluster finding via quantum algorithms. As a main conclusion of the work, via the utilization of quantum computers and quantum machine learning, an exponential speed-up can be reached over classical algorithms.

In \cite{ref18}, the authors defined a method for the analysis of an unknown quantum state. The authors showed that it is possible to perform “quantum principal component analysis” by creating quantum coherence among different copies, and the relevant attributes can be revealed exponentially faster than it is possible by any existing algorithm.

In \cite{ref13}, the authors studied the application of a quantum support vector machine in Big Data classification. The authors showed that a quantum version of the support vector machine (optimized binary classifier) can be implemented on a quantum computer. As the work concluded, the complexity of the quantum algorithm is only logarithmic in the size of the vectors and the number of training examples that provides a significant advantage over classical support machines.

In \cite{ref16}, the problem of quantum-based analysis of big data sets is studied by the authors. As the authors concluded, the proposed quantum algorithms provide an exponential speedup over classical algorithms for topological data analysis.

The problem of quantum generative adversarial learning is studied in \cite{ref35}. In generative adversarial networks a generator entity creates statistics for data that mimics those of a valid data set, and a discriminator unit distinguishes between the valid and non-valid data. As a main conclusion of the work, a quantum computer allows us to realize quantum adversarial networks with an exponential advantage over classical adversarial networks.

In \cite{ref31}, super-polynomial and exponential improvements for quantum-enhanced reinforcement learning are studied. 

In \cite{ref32}, the authors proposed strategies for quantum computing molecular energies using the unitary coupled cluster ansatz. 

The authors of \cite{ref33} provided demonstrations of quantum advantage in machine learning problems.

In \cite{ref34}, the authors study the subject of quantum speedup in machine learning. As a particular problem, the work focuses on finding Boolean functions for classification tasks.
\section{System Model}
\label{sec2}
\subsection{Gate-Model Quantum Neural Network}
\begin{definition}
A ${\rm QNN}_{QG} $ is a quantum neural network (${\rm QNN}$) implemented on a gate-model quantum computer with a quantum gate structure $QG$. It contains quantum links between the unitaries and classical links for the propagation of classical side information. In a ${\rm QNN}_{QG} $, all quantum information propagates forward from the input to the output, while classical side information can propagate arbitrarily (forward and backward) in the network. In a ${\rm QNN}_{QG} $, there is no available side information about the previous running sequences of the structure.
\end{definition}
Using the framework of \cite{ref9}, a ${\rm QNN}_{QG} $ is formulated by a collection of $L$ unitary gates, such that an $i$-th, $i=1,\ldots ,L$ unitary gate $U_{i} \left(\theta _{i} \right)$ is
\begin{equation} \label{1)} 
U_{i} \left(\theta _{i} \right)=\exp \left(-i\theta _{i} P\right),                                                               
\end{equation} 
where $P$ is a generalized Pauli operator formulated by a tensor product of Pauli operators $\left\{X,Y,Z\right\}$, while $\theta _{i} $ is referred to as the gate parameter associated with $U_{i} \left(\theta _{i} \right)$.

In ${\rm QNN}_{QG} $, a given unitary gate $U_{i} \left(\theta _{i} \right)$ sequentially acts on the output of the previous unitary gate $U_{i-1} \left(\theta _{i-1} \right)$, without any nonlinearities \cite{ref9}. The classical side information of ${\rm QNN}_{QG} $ is used in calculations related to error derivation and gradient computations, such that side information can propagate arbitrarily in the network structure.

The sequential application of the $L$ unitaries formulates a unitary operator $U(\vec{\theta })$ as
\begin{equation} \label{ZEqnNum335090} 
U(\vec{\theta })=U_{L} \left(\theta _{L} \right)U_{L-1} \left(\theta _{L-1} \right)\ldots U_{1} \left(\theta _{1} \right),                                                
\end{equation} 
where $U_{i} \left(\theta _{i} \right)$ identifies an $i$-th unitary gate, and $\vec{\theta }$ is the gate parameter vector
\begin{equation} \label{ZEqnNum837426} 
\vec{\theta }=\left(\theta _{1} ,\ldots ,\theta _{L-1} ,\theta _{L} \right)^{T}.
\end{equation} 
At \eqref{ZEqnNum335090}, the evolution of the system of ${\rm QNN}_{QG} $ for a particular input system ${\left| \psi ,\varphi  \right\rangle} $ is
\begin{equation} \label{4)} 
{\left| Y \right\rangle} =U(\vec{\theta }){\left| \psi  \right\rangle} {\left| \varphi  \right\rangle} =U(\vec{\theta }){\left| z \right\rangle} {\left| 1 \right\rangle} =U(\vec{\theta }){\left| z,1 \right\rangle} ,                                            
\end{equation} 
where ${\left| Y \right\rangle} $ is the $\left(n+1\right)$-length output quantum system, and ${\left| \psi  \right\rangle} ={\left| z \right\rangle} $ is a computational basis state, where $z$ is an $n$-length string  
\begin{equation} \label{5)} 
z=z_{1} z_{2} \ldots z_{n} , 
\end{equation} 
where each $z_{i} $ represents a classical bit with values 
\begin{equation} \label{6)} 
z_{i} \in \left\{-1,1\right\},                                                                 
\end{equation} 
while the $\left(n+1\right)$-th quantum state is initialized as 
\begin{equation} \label{ZEqnNum360835} 
{\left| \varphi  \right\rangle} ={\left| 1 \right\rangle} ,                                                                    
\end{equation} 
and is referred to as the readout quantum state.  

\subsection{Objective Function}

The $f(\vec{\theta })$ objective function subject to minimization is defined for a ${\rm QNN}_{QG} $ as
\begin{equation} \label{ZEqnNum830871} 
f(\vec{\theta })={\langle \vec{\theta } |{\rm {\mathcal L}}(x_{0} ,\tilde{l}(z))|\vec{\theta } \rangle},  
\end{equation} 
where ${\rm {\mathcal L}}(x_{0} ,\tilde{l}(z))$ is the loss function \cite{ref9}, defined as
\begin{equation} \label{ZEqnNum125883} 
{\rm {\mathcal L}}(x_{0} ,\tilde{l}(z))=1-l\left(z\right)\tilde{l}\left(z\right),                                                    
\end{equation} 
where $\tilde{l}\left(z\right)$ is the predicted value of the binary label
\begin{equation}
l\left(z\right)\in \left\{-1,1\right\}
\end{equation}
of the string $z$, defined as \cite{ref9}
\begin{equation} \label{ZEqnNum174011} 
\tilde{l}(z)={\langle z,1 \mathrel{| \vphantom{z,1 (U(\vec{\theta }))^{\dag } Y_{n+1} U(\vec{\theta })|z,1}\kern-\nulldelimiterspace} (U(\vec{\theta }))^{\dag } Y_{n+1} U(\vec{\theta })|z,1 \rangle} ,                                           
\end{equation} 
where $Y_{n+1} \in \left\{-1,1\right\}$ is a measured Pauli operator on the readout quantum state \eqref{ZEqnNum360835}, while $x_{0} $ is as 
\begin{equation} \label{ZEqnNum338221} 
x_{0} ={\left| z,1 \right\rangle }.
\end{equation} 

The $\tilde{l}$ predicted value in \eqref{ZEqnNum174011} is a real number between $-1$ and $1$, while the label $l\left(z\right)$ and $Y_{n+1} $ are real numbers $-1$ or $1$. Precisely, the $\tilde{l}$ predicted value as given in \eqref{ZEqnNum174011} represents an average of several measurement outcomes if $Y_{n+1} $ is measured via $R$ output system instances ${\left| Y \right\rangle} ^{\left(r\right)} $-s, $r=1,\ldots ,R$ \cite{ref9}.

The learning problem for a ${\rm QNN}_{QG} $ is, therefore, as follows. At an ${\rm {\mathcal S}}_{T} $ training set formulated via $R$ input strings and labels 
\begin{equation} \label{11)} 
{\rm {\mathcal S}}_{T} =\left\{z^{(r)},l\left(z^{(r)} \right),r=1,\ldots ,R\right\} ,                                                   
\end{equation}
where $r$ refers to the $r$-th measurement round and $R$ is the total number of measurement rounds, the goal is therefore to find the gate parameters \eqref{ZEqnNum837426} of the $L$ unitaries of ${\rm QNN}_{QG} $, such that $f(\vec{\theta })$ in \eqref{ZEqnNum830871} is minimal.

\subsection{Recurrent Gate-Model Quantum Neural Network}
\begin{definition}
An ${\rm RQNN}_{QG} $ is a ${\rm QNN}$ implemented on a gate-model quantum computer with a quantum gate structure $QG$, such that the connections of ${\rm RQNN}_{QG} $ form a directed graph along a sequence. It contains quantum links between the unitaries and classical links for the propagation of classical side information. In an ${\rm RQNN}_{QG} $, all quantum information propagates forward, while classical side information can propagate only backward direction. In an ${\rm RQNN}_{QG} $, side information is available about the previous running sequences of the structure.
\end{definition}
The classical side information of ${\rm RQNN}_{QG} $ is used in error derivation and gradient computations, such that side information can propagate only in backward directions. Similar to the ${\rm QNN}_{QG} $ case, in an ${\rm RQNN}_{QG} $, a given $i$-th unitary $U_{i} \left(\theta _{i} \right)$ acts on the output of the previous unitary $U_{i-1} \left(\theta _{i-1} \right)$. Thus, the quantum evolution of the ${\rm RQNN}_{QG} $ contains no nonlinearities \cite{ref9}. As follows, for an ${\rm RQNN}_{QG} $ network, the objective function can be similarly defined as given in \eqref{ZEqnNum830871}. On the other hand, the structural differences between ${\rm QNN}_{QG} $ and ${\rm RQNN}_{QG} $ allows the characterization of different computational models for the description of the learning problem. The structural differences also lead to various optimal learning methods for the ${\rm QNN}_{QG} $ and ${\rm RQNN}_{QG} $ structures as it will be revealed in \sref{sec4} and \sref{sec5}.
\subsection{Comparative Representation}
For a simple graphical representation, the schematic models of a ${\rm QNN}_{QG} $ and ${\rm RQNN}_{QG} $ for an $\left( r-1 \right)$-th and $r$-th measurement rounds are compared in \fref{figA1}. The $(n+1)$-length input systems are depicted by $\left| {{\psi }_{r-1}} \right\rangle \left| 1 \right\rangle $ and $\left| {{\psi }_{r}} \right\rangle \left| 1 \right\rangle $, while the output systems are denoted by $\left| {{Y}_{r-1}} \right\rangle $ and $\left| {{Y}_{r}} \right\rangle $. The result of the $M$ measurement operator in the $\left( r-1 \right)$-th and $r$-th measurement rounds are denoted by $Y_{n+1}^{\left( r-1 \right)}$ and $Y_{n+1}^{\left( r\right)}$.
 In \fref{figA1}(a), structure of a ${\rm QNN}_{QG} $ is depicted for an $\left( r-1 \right)$-th and $r$-th measurement round. In \fref{figA1}(b), the structure of a ${\rm RQNN}_{QG} $ is illustrated. In a ${\rm QNN}_{QG} $, side information is not available about the previous, $\left( r-1 \right)$-th measurement round in a particular $r$-th  measurement round. For an ${\rm RQNN}_{QG} $, side information is available about the $\left( r-1 \right)$-th measurement round (depicted by the dashed gray arrows) in a particular $r$-th measurement round. The side information in the ${\rm RQNN}_{QG} $ setting refer to information about the gate-parameters and the measurement results of the $\left( r-1 \right)$-th measurement round. 

 \begin{center}
\begin{figure*}[h!]
%\vspace{-0.4cm}
\begin{center}
\includegraphics[angle = 0,width=1\linewidth]{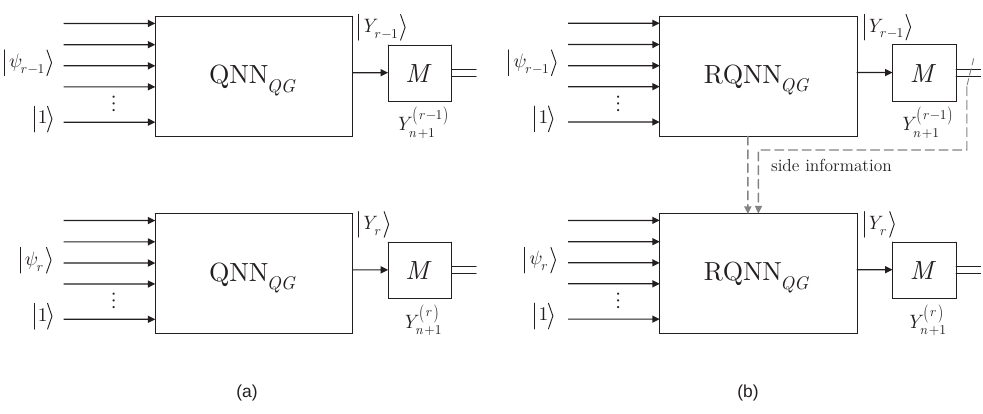}
\caption{Schematic representation of a ${\rm QNN}_{QG} $ and ${\rm RQNN}_{QG}$ in an $(r-1)$-th and $r$-th measurement rounds. The $(n+1)$-length input systems of the $\left( r-1 \right)$ and $r$-th measurement rounds are depicted by $\left| {{\psi }_{r-1}} \right\rangle \left| 1 \right\rangle $ and $\left| {{\psi }_{r}} \right\rangle \left| 1 \right\rangle $, the output systems are $\left| {{Y}_{r-1}} \right\rangle $ and $\left| {{Y}_{r}} \right\rangle $. The result of the $M$ measurement operator in the $\left( r-1 \right)$-th and $r$-th measurement rounds are denoted by $Y_{n+1}^{\left( r-1 \right)}$ and $Y_{n+1}^{\left( r\right)}$. (a): In a ${\rm QNN}_{QG} $, for an $r$-th measurement round, side information is not available about the previous $\left( r-1 \right)$-th measurement round. (b): In a ${\rm RQNN}_{QG} $, side information is available about the previous $\left( r-1 \right)$-th measurement round in an $r$-th round.} 
 \label{figA1}
 \end{center}
\end{figure*}
\end{center}

\subsection{Parameterization}
\label{sec3}
\subsubsection{Constraint Machines}
The tasks of machine learning can be modeled via its mathematical framework and the constraints of the environment \cite{ref1, ref2, ref3}. A ${\rm {\mathcal C}}$ constraint machine is a learning machine working with constraints \cite{ref23}. A constraint machine can be formulated by a particular function $f$ or via some elements of a functional space ${\rm {\mathcal F}}$. The constraints model the attributes of the environment of ${\rm {\mathcal C}}$. 

The learning problem of a ${\rm {\mathcal C}}$ constraint machine can be represented via a ${\rm {\mathcal G}}=\left(V,S\right)$ environmental graph \cite{ref23,ref24,ref25,ref26,ref27}. The ${\rm {\mathcal G}}$ environmental graph is a directed acyclic graph (DAG), with a set $V$ of vertexes and a set $S$ of arcs. The vertexes of ${\rm {\mathcal G}}$ model associated features, while the arcs between the vertexes describe the relations of the vertexes. 

The ${\rm {\mathcal G}}$ environmental graph formalizes factual knowledge via modeling the relations among the elements of the environment \cite{ref23}. In the environmental graph representation, the ${\rm {\mathcal C}}$ constraint machine has to decide based on the information associated with the vertexes of the graph.

For any vertex $v$ of $V$, a perceptual space element $x$, and its identifier $\left\langle x\right\rangle $ that addresses $x$ in the computational model can be defined as a pair
\begin{equation} \label{12)} 
\left(\left\langle x\right\rangle ,x\right),                                                                     
\end{equation} 
where $x\in {\rm {\mathcal X}}$ is an element (vector) of the perceptual space ${\rm {\mathcal X}}\subset {\rm {\mathbb{C}}}^{d} $. Assuming that features are missing, the $\diamondsuit $ symbol can be used. Therefore, ${\rm {\mathcal X}}$ is initialized as ${\rm {\mathcal X}}_{0}$,
\begin{equation} \label{xnull} 
{\rm {\mathcal X}}_{0} ={\rm {\mathcal X}}\bigcup \left\{\diamondsuit \right\}.                                                                 
\end{equation} 
The environment is populated by individuals, and the ${\rm {\mathcal I}}$ individual space is defined via $V$ and ${\rm {\mathcal X}}_{0} $ as 
\begin{equation} \label{14)} 
{\rm {\mathcal I}}=V\times {\rm {\mathcal X}}_{0} ,                                                                 
\end{equation} 
such that the existing features are associated with a subset $\tilde{V}$ of $V$.

The features can be associated with the $\left\langle x\right\rangle $ identifier via a $f_{{\rm {\mathcal P}}} $ perceptual map as
\begin{equation} \label{15)} 
f_{{\rm {\mathcal P}}} :\tilde{V}\to {\rm {\mathcal X}:\; }:x=f_{{\rm {\mathcal P}}} \left(v\right).                                                       
\end{equation} 
If the condition 
\begin{equation} \label{16)} 
\forall v\in (V\backslash \tilde{V}):x=f_{{\rm {\mathcal P}}} (v)=\diamondsuit  
\end{equation} 
holds, then $f_{{\rm {\mathcal P}}} $ is yielded as
\begin{equation} \label{17)} 
f_{{\rm {\mathcal P}}} :V\to {\rm {\mathcal X}:\; }:x=f_{{\rm {\mathcal P}}} \left(v\right).                                                      
\end{equation} 
A given individual $\iota \in {\rm {\mathcal I}}$ is defined as a feature vector $x\in {\rm {\mathcal X}}$. An $\iota \in {\rm {\mathcal I}}$ individual of the individual space ${\rm {\mathcal I}}$ is defined as
\begin{equation} \label{ZEqnNum107192} 
\iota =\Upsilon x+\neg \Upsilon v,                                                                  
\end{equation} 
where $+$ is the sum operator in ${\rm {\mathbb{C}}}^{d} $, $\neg $ is the negation operator, while $\Upsilon $ is a constraint as
\begin{equation} \label{19)} 
\Upsilon :(v\in \tilde{V})\vee \left(x\in {\rm {\mathcal X}}\backslash \tilde{{\rm {\mathcal X}}_{0}}  \right).                                                     
\end{equation} 
where ${\rm {\mathcal X}}_{0}$ is given in \eqref{xnull}.
Thus, from \eqref{ZEqnNum107192}, an individual $\iota $ is a feature vector $x$ of ${\rm {\mathcal X}}$ or a vertex $v$ of ${\rm {\mathcal G}}$. 

Let $\iota ^{*} \in {\rm {\mathcal I}}$ be a specific individual, and let $f$ be an agent represented by the function $f:{\rm {\mathcal I}}\to {\rm {\mathbb{C}}}^{n} $. Then, at a given environmental graph ${\rm {\mathcal G}}$, the ${\rm {\mathcal C}}$ constraint machine is defined via function $f$ as a machine in which the learning and inference are represented via enforcing procedures on constraints $C_{\iota ^{*} } $ and $C_{\iota } $, such that for a ${\rm {\mathcal C}}$ constraint machine the learning procedure requires the satisfaction of the constraints over all ${\rm {\mathcal I}}^{*} $, while in the inference the satisfaction of the constraint is enforced over the given $\iota ^{*} \in {{{\mathcal I}}} $ \cite{ref23}, by theory. Thus, ${\rm {\mathcal C}}$ is defined in a formalized manner, as 
\begin{equation} \label{ZEqnNum864288} 
{\rm {\mathcal C}}\equiv \left\{\begin{array}{l} {C_{\iota } :\forall \iota \in {\tilde{{\mathcal I}}} :\chi \left(v,f\left(\iota \right)\right)=0,} \\ {C_{\iota ^{*} } :\iota ^{*} \in {\rm {\mathcal I}\backslash \tilde{{\mathcal I}}} :\chi \left(v^{*} ,f^{*} \left(\iota ^{*} \right)\right)=0,} \end{array}\right. 
\end{equation} 
where $\tilde{{\mathcal I}}$ is a subset of ${\mathcal I}$, $\iota ^{*} $ refers to a specific individual, vertex or function, $\chi \left(\cdot \right)$ is a compact constraint function, while $v^{*} $ and $f^{*} \left(\iota ^{*} \right)$ refer to the vertex and function at $\iota ^{*} $, respectively.

\subsubsection{Calculus of Variations}
Some elements from the calculus of variations \cite{ref39,ref40} are utilized in the learning optimization procedure.

\paragraph{Euler-Lagrange Equations}
The Euler-Lagrange equations are second-order partial differential equations with solution functions. These equations are useful in optimization problems since they have a differentiable functional that is stationary at the local maxima and minima \cite{ref39}. As a corollary, they can be also used in the problems of machine learning.

\paragraph{Hessian Matrix}
A Hessian matrix $\mathbf{H}$ is a square matrix of second-order partial derivatives of a scalar-valued function, or scalar field \cite{ref39}. In theory, it describes the local curvature of a function of many variables. In a machine-learning setting, it is a useful tool to derive some attributes and critical points of loss functions.

\section{Constraint-based Computational Model}
\label{sec4}
In this section, we derive the computational models of the ${\rm QNN}_{QG} $ and ${\rm RQNN}_{QG} $ structures.

\subsection{Environmental Graph of a Gate-Model Quantum Neural Network}

\begin{proposition}
The ${\rm {\mathcal G}}_{{\rm QNN}_{QG} } =\left(V,S\right)$ environmental graph of a ${\rm QNN}_{QG} $ is a DAG, where $V$ is a set of vertexes, in our setting defined as
\begin{equation} \label{21)} 
V={\rm {\mathcal S}}_{in} \bigcup {\rm {\mathcal U}}\bigcup {\rm {\mathcal Y}},                                                            
\end{equation} 
where ${\rm {\mathcal S}}_{in} $ is the input space, ${\rm {\mathcal U}}$ is the space of unitaries, ${\rm {\mathcal Y}}$ is the output space, and $S$ is a set of arcs.
\end{proposition}

Let ${\rm {\mathcal G}}_{{\rm QNN}_{QG} } $ be an environmental graph of ${\rm QNN}_{QG} $, and let $v_{U_{i}} $ be a vertex, such that $v_{U_{i}} \in V$ is related to the unitary $U_{i} \left(\theta _{i} \right)$, where index $i=0$ is associated with the ${\left| z,1 \right\rangle} $ input system with vertex $v_{0} $. Then, let $v_{U_{i}} $ and $v_{U_{j}} $ be connected vertices via directed arc $s_{ij} $, $s_{ij} \in S$, such that a particular $\theta _{ij} $ gate parameter is associated with the forward directed arc\footnote{The notation $U_{j}(\theta _{ij})$ refers to the selection of $\theta _{j}$ for the unitary $U_j$ to realize the operation $U_i\left( {{\theta }_{i}} \right)U_j\left( {{\theta }_{j}} \right)$, i.e., the application of $U_j\left( {{\theta }_{j}} \right)$ on the output of $U_i\left( {{\theta }_{i}} \right)$ at a particular gate parameter $\theta _{j}$.}, as
\begin{equation} \label{22)} 
\theta _{ij} =\theta _{j} ,                                                                        
\end{equation} 
such that arc $s_{0j} $ is associated with $\theta _{0j} =\theta _{j} $.

Then a given state $x_{U_{i} \left(\theta _{i} \right)} $ of ${\rm {\mathcal X}}$ associated with $U_{i} \left(\theta _{i} \right)$ is defined as
\begin{equation} \label{23)} 
x_{U_{i} \left(\theta _{i} \right)} =v_{U_{i}} +a_{U_{i} \left(\theta _{i} \right)} , 
\end{equation}
where $v_{U_{i}} $ is a label for unitary $U_{i}$ in the environmental graph ${\rm {\mathcal G}}_{{\rm QNN}_{QG} } $ (serves as an identifier in the computational structure of \eqref{23)}), while parameter $a_{U_{i} \left(\theta _{i} \right)} $ is defined for a $U_{i} \left(\theta _{i} \right)$ as
\begin{equation} \label{25)} 
a_{U_{i} \left(\theta _{i} \right)} =\sum _{h\in \Xi \left(i\right)}U_{i} \left(\theta _{hi} \right)x_{U_{h} \left(\theta _{h} \right)} +b_{U_{i} \left(\theta _{i} \right)}  , 
\end{equation} 
where $\Xi \left(i\right)$ refers to the parent set of $v_{U_{i}} $,  $U_{i}(\theta _{hi})$ refers to the selection of $\theta _{i}$ for unitary $U_i$ for a particular input from $U_h(\theta _{h})$, while $b_{U_{i} \left(\theta _{i} \right)} $ is the bias relative to $v_{U_{i}} $.

Applying a $f_{\angle }$ topological ordering function on ${\rm {\mathcal G}}_{{\rm QNN}_{QG} } $ yields an ordered graph structure $f_{\angle } ({\rm {\mathcal G}}_{{\rm QNN}_{QG} })$ of the $L$ unitaries. Thus, a given output ${\left| Y \right\rangle} $ of ${\rm QNN}_{QG} $ can be rewritten in a compact form as
\begin{equation} \label{26)} 
{\left| Y \right\rangle} =U(\vec{\theta })x_{0} =\left(U_{L} \left(\theta _{L} \right)U_{L-1} \left(\theta _{L-1} \right)\ldots U_{1} \left(\theta _{1} \right)x_{0} \right),                                  
\end{equation} 
where the term $x_{0} \in {\rm {\mathcal S}}_{in} $ is associated with the input system as defined in \eqref{ZEqnNum338221}.

A particular state $x_{U_{l} \left(\theta _{l} \right)} $, $l=1,\ldots ,L$ is evaluated in function of $x_{U_{l-1} \left(\theta _{l-1} \right)} $ as
\begin{equation} \label{27)} 
x_{U_{l} \left(\theta _{l} \right)} =U_{l} \left(\theta _{l} \right)x_{U_{l-1} \left(\theta _{l-1} \right)} .                                                              
\end{equation} 

The environmental and ordered graphs of a gate-model quantum neural network are illustrated in \fref{fig1}. In \fref{fig1}(a) the ${\rm {\mathcal G}}_{{\rm QNN}_{QG} } $ environmental graph of a ${\rm QNN}_{QG} $ is depicted, and the ordered graph $f_{\angle } ({\rm {\mathcal G}}_{{\rm QNN}_{QG} })$ is shown in \fref{fig1}(b). 

 \begin{center}
\begin{figure*}[h!]
%\vspace{-0.4cm}
\begin{center}
\includegraphics[angle = 0,width=1\linewidth]{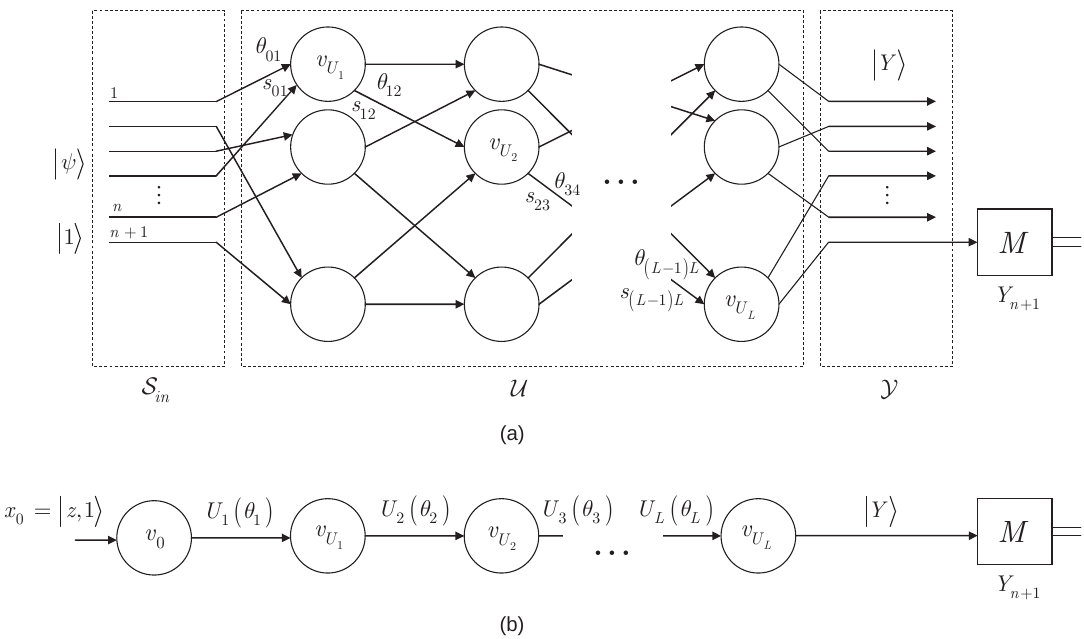}
\caption{(a): The ${\rm {\mathcal G}}_{{\rm QNN}_{QG} } $ environmental graph of a ${\rm QNN}_{QG} $, with $L$ unitaries. The input state of the ${\rm QNN}_{QG} $ is ${\left| \psi ,1 \right\rangle} $. A unitary $U_{i} \left(\theta _{i} \right)$ is represented by a vertex $v_{U_{i}} $ in ${\rm {\mathcal G}}_{{\rm QNN}_{QG} } $. The vertices $v_{U_{i}} $ and $v_{U_{j}} $ of unitaries $U_{i} \left(\theta _{i} \right)$ and $U_{j} \left(\theta _{j} \right)$ are connected by directed arcs $s_{ij} $. The gate parameters $\theta _{ij} =\theta _{j} $ are associated with $s_{ij} $, while ${\rm {\mathcal S}}_{in} $ is the input space, ${\rm {\mathcal U}}$ is the space of $L$ unitaries, and ${\rm {\mathcal Y}}$ is the output space. Operator $M$ is a measurement on the $\left(n+1\right)$-th state (readout quantum state), and $Y_{n+1} $ is a Pauli operator measured on the readout state (classical links are not depicted) (b): The compacted ordered graph $f_{\angle } ({\rm {\mathcal G}}_{{\rm QNN}_{QG} })$ of ${\rm QNN}_{QG} $. The output is ${\left| Y \right\rangle} =U(\vec{\theta })x_{0} $, where $x_{0} ={\left| z,1 \right\rangle} $ and $U(\vec{\theta })=\prod _{l=1}^{L}U_{l} \left(\theta _{l} \right) $(classical links are not depicted).} 
 \label{fig1}
 \end{center}
\end{figure*}
\end{center}

\subsection{Computational Model of Gate-Model Quantum Neural Networks}
\begin{theorem}
The computational model of a ${\rm QNN}_{QG} $ is a ${\rm {\mathcal C}}\left({\rm QNN}_{QG} \right)$ constraint machine with linear transition functions $f_{T} \left({\rm QNN}_{QG} \right)$. 
\end{theorem}
\begin{proof} 
Let ${\rm {\mathcal G}}\left({\rm QNN}_{QG} \right)=\left(V,S\right)$ be the environmental graph of a ${\rm QNN}_{QG} $, and assume that the number of types of the vertexes is $p$. Then, the vertex set $V$ can be expressed as a collection 
\begin{equation} \label{28)} 
V=\bigcup _{i=1}^{p}V_{i}  , 
\end{equation} 
where $V_{i}$ identifies a set of vertexes, $p$ is the total number of the $V_{i}$ sets, such that $V_{i} \bigcap V_{j}  =\emptyset ,$ if only $i\ne j$ \cite{ref23}.
For a $v \in V_{i}$ vertex from set $V_{i}$, an ${{f}_{T}}:\mathbb{C}^{\dim _{in}} \to \mathbb{C}^{\dim _{out}}$ transition function \cite{ref23} can be defined as
\begin{equation} \label{ZEqnNum801212} 
f_{T} :{\rm {\mathcal Z}}_{{V_{i}}}^{\left|\Gamma \left(v\right)\right|} \times {\rm {\mathcal X}}_{{V_{i}}} \to {\rm {\mathcal Z}}_{{V_{i}}} :\left(\gamma _{\Gamma \left(v\right)} ,x_{v} \right)\to f_{T} \left(\gamma _{\Gamma \left(v\right)} ,x_{v} \right),                            
\end{equation} 
where ${\rm {\mathcal X}}_{{V_{i}}} $ is the perceptual space ${\rm {\mathcal X}}$ of ${V_{i}}$, ${\rm {\mathcal X}}_{{V_{i}}} \subset {\rm {\mathbb{C}}}^{\dim ({\rm {\mathcal X}}_{{V_{i}}}) } $; $\dim ({\rm {\mathcal X}}_{{V_{i}}}) $ is the dimension of the space ${\rm {\mathcal X}}_{{V_{i}}} $; $x_{v} $ is an element of ${\rm {\mathcal X}}_{{V_{i}}} $; $x_{v} \in {\rm {\mathcal X}}_{{V_{i}}} $ associated with a unitary $U_{v} \left(\theta _{v} \right)$; ${\rm {\mathcal Z}}$ is the state space, ${\rm {\mathcal Z}}_{{V_{i}}} $ is the state space of $V_{i}$, ${\rm {\mathcal Z}}_{{V_{i}}} \subset {\rm {\mathbb{C}}}^{\dim ({\rm {\mathcal Z}}_{{V_{i}}})} $, $\dim ({\rm {\mathcal Z}}_{{V_{i}}}) $ is the dimension of the space ${\rm {\mathcal Z}}_{{V_{i}}} $; $\Gamma \left(v\right)$ refers to the children set of $v$; $\left|\Gamma \left(v\right)\right|$ is the cardinality of set $\Gamma \left(v\right)$; $\gamma \in {\mathcal Z} $ is a state variable in the state space ${\rm {\mathcal Z}}$ that serves as side information to process the $v$ vertices of $V$ in ${\rm {\mathcal G}}\left({\rm QNN}_{QG} \right)$, while $\gamma_{\Gamma(v)} \in {\rm {\mathcal Z}}_{V_i}^{|\Gamma(v)|} \subset {\mathbb{C}}^{ |\Gamma(v)| }$ and $\gamma_{\Gamma(v)} = (\gamma_{\Gamma(v), 1}, \cdots , \gamma_{\Gamma(v), |\Gamma(v)|}) $, by theory \cite{ref23,ref27}. Thus, the $f_T$ transition function in \eqref{ZEqnNum801212} is a complex-valued function that maps an input pair $(\gamma , x)$ from the space of $\mathcal{X} \times \mathcal{Z}$ to the state space $\mathcal{Z}$.

Similarly, for any ${V_{i}}$, an $f_{O}:\mathbb{C}^{\dim _{in}} \to \mathbb{C}^{\dim _{out}}$ output function \cite{ref23} can be defined as
\begin{equation} \label{ZEqnNum904058} 
F_{O} :{\rm {\mathcal Z}}_{{V_{i}}} \times {\rm {\mathcal X}}_{{V_{i}}} \to {\rm {\mathcal Y}}_{{V_{i}}}:\left(\gamma _{v} ,x_{v} \right)\to F_{O} \left(\gamma _{v} ,x_{v} \right), 
\end{equation} 
where ${\rm {\mathcal Y}}_{{V_{i}}} $ is the output space ${\rm {\mathcal Y}}$, and $\gamma _{v} $ is a state variable associated with $v$, $\gamma _{v} \in {\rm {\mathcal Z}}_{{V_{i}}} $, such that $\gamma _{v} =\gamma _{0} $ if $\Gamma \left(v\right)=\emptyset $. The $ f_{O}$ output function in \eqref{ZEqnNum904058} is therefore a complex-valued function that maps an input pair $(\gamma , x)$ from the space of $\mathcal{X} \times \mathcal{Z}$ to the output space $\mathcal{Y}$.

From \eqref{ZEqnNum801212} and \eqref{ZEqnNum904058}, it follows that for any ${V_{i}}$, there exists the $\phi \left({V_{i}}\right)$ associated function-pair as
\begin{equation} \label{31)} 
\phi \left({V_{i}}\right)=\left(f_{T} ,F_{O} \right).                                                      
\end{equation} 
Let us specify the generalized functions of \eqref{ZEqnNum801212} and \eqref{ZEqnNum904058} for a ${\rm QNN}_{QG} $.

Let $U(\vec{\theta })$ of ${\rm QNN}_{QG} $ be defined as given in \eqref{ZEqnNum335090}. Since in ${\rm QNN}_{QG} $, a given $i$-th unitary $U_{i} \left(\theta _{i} \right)$ acts on the output of the previous unitary $U_{i-1} \left(\theta _{i-1} \right)$, the network contains no nonlinearities \cite{ref9}. As a corollary, the state transition function $f_{T} \left({\rm QNN}_{QG} \right)$ in \eqref{ZEqnNum801212} is also linear for a ${\rm QNN}_{QG} $.

Let $\left| \gamma _{v}  \right\rangle $ be the quantum state associated with $\gamma _{v}$ state variable of a given $v$. Then, the constraints on the transition function and output function of a ${\rm QNN}_{QG} $ can be evaluated as follows. 

Let $f_{T} \left({\rm QNN}_{QG} \right)$ be the transition function of a ${\rm QNN}_{QG} $ defined for a given $v\in V$ of ${\rm {\mathcal G}}\left({\rm QNN}_{QG} \right)$ via \eqref{ZEqnNum801212} as
\begin{equation} \label{ZEqnNum916033} 
f_{T} \left({\rm QNN}_{QG} \right):\left(\gamma _{\Gamma \left(v\right)} ,x_{v} \right)\to f_{T} \left(\gamma _{\Gamma \left(v\right)} ,x_{v} \right).                                       
\end{equation} 
The $F_{O} \left({\rm QNN}_{QG} \right)$ output function of a ${\rm QNN}_{QG} $ for a given $v$ of ${\rm {\mathcal G}}\left({\rm QNN}_{QG} \right)$ via \eqref{ZEqnNum904058} is
\begin{equation} \label{ZEqnNum988516} 
F_{O} \left({\rm QNN}_{QG} \right):\left(\gamma _{v} ,x_{v} \right)\to F_{O} \left(\gamma _{v} ,x_{v} \right).                                            
\end{equation} 
Since $f_{T} \left({\rm QNN}_{QG} \right)$ in \eqref{ZEqnNum916033} and $F_{O} \left({\rm QNN}_{QG} \right)$ in \eqref{ZEqnNum988516} correspond with the data flow computational scheme of a ${\rm QNN}_{QG} $ with linear transition functions, \eqref{ZEqnNum916033} and \eqref{ZEqnNum988516} represent an expression of the constraints of ${\rm QNN}_{QG} $. These statements can be formulated in a compact form.

Let $\zeta _{v} $ be a constraint on $f_{T} \left({\rm QNN}_{QG} \right)$ of ${\rm QNN}_{QG} $ as 
\begin{equation} \label{ZEqnNum270596} 
\zeta _{v} :{\left| \gamma _{v}  \right\rangle} -f_{T} \left({\rm QNN}_{QG} \right)=0. 
\end{equation} 
Thus, the $f_{T} \left({\rm QNN}_{QG} \right)$ transition function is constrained as
\begin{equation} \label{37)} 
f_{T} \left({\rm QNN}_{QG} \right)={\left| \gamma _{v}  \right\rangle} .                                                            
\end{equation} 
With respect to the output function, let $\varphi _{v} $ be a constraint on $F_{O} \left({\rm QNN}_{QG} \right)$ of ${\rm QNN}_{QG} $ as
\begin{equation} \label{ZEqnNum506759} 
\varphi _{v} :\wp _{v} \circ F_{O} \left({\rm QNN}_{QG} \right)=0,                                                        
\end{equation} 
where $\circ $ is the composition operator, such that $\left(f\circ g\right)\left(x\right)=f\left(g\left(x\right)\right)$, $\wp _{v}$ is therefore another constraint as $\wp _{v}(F_{O} \left({\rm QNN}_{QG} \right))=0$.

Then let $\pi _{v} $ be a compact constraint on $f_{T} \left({\rm QNN}_{QG} \right)$ and $F_{O} \left({\rm QNN}_{QG} \right)$ defined via constraints \eqref{ZEqnNum270596} and \eqref{ZEqnNum506759} as
\begin{equation} \label{ZEqnNum734924} 
\begin{split}
\pi _{v} \left(f_{T} \left({\rm QNN}_{QG} \right),F_{O} \left({\rm QNN}_{QG} \right)\right)\\
=&\sum _{v\in V}(\zeta _{v} +\varphi _{v})-2\left|V\right| \\
=&\sum _{v\in V}\left((\left| \gamma _{v}  \right\rangle-(f_{T} \left({\rm QNN}_{QG} \right)) +\wp _{v}(F_{O} \left({\rm QNN}_{QG} \right)) \right)-2\left|V\right| \\
 \ne & 0.                                   
\end{split}
\end{equation} 
Since it can be verified that a learning machine that enforces the constraint in \eqref{ZEqnNum734924}, is in fact a constraint machine. As a corollary, the constraints \eqref{ZEqnNum916033} and \eqref{ZEqnNum988516}, along with the compact constraint \eqref{ZEqnNum734924}, define a ${\rm {\mathcal C}}\left({\rm QNN}_{QG} \right)$ constraint machine for a ${\rm QNN}_{QG} $ with linear functions $f_{T} \left({\rm QNN}_{QG} \right)$ and $F_{O} \left({\rm QNN}_{QG} \right)$.
\end{proof}
\subsection{Diffusion Machine}
Let ${\rm {\mathcal C}}$ be the constraint machine with linear transition function $f_{T} \left(\gamma _{\Gamma \left(v\right)} ,x_{v} \right)$, and let $\S _{v}$ be a state variable such that $\forall v\in V$
\begin{equation} \label{add1} 
\S _{v}-f_{T} \left(\gamma _{\Gamma \left(v\right)} ,x_{v} \right) =0, 
\end{equation} 
and let $F_{O} \left(\gamma _{v} ,x_{v} \right)$ be the output function of ${\rm {\mathcal C}}$, such that $\forall v\in V$
\begin{equation} \label{add2} 
c_{v} \circ F_{O} \left(\gamma _{v} ,x_{v} \right)=0,
\end{equation}
where $c_{v}$ is a constraint.

Then, the ${\rm {\mathcal C}}$ constraint machine is a ${\rm {\mathcal D}}$ diffusion machine \cite{ref23}, if only ${\rm {\mathcal C}}$ enforces the constraint ${{C}_{\mathcal{D}}}$, as 
\begin{equation} \label{xeq} 
{{C}_{\mathcal{D}}}: \sum _{v\in V}\left((\S _{v}-f_{T} \left(\gamma _{\Gamma \left(v\right)} ,x_{v} \right) =0) +(c_{v} \circ F_{O} \left(\gamma _{v} ,x_{v} \right)=0) \right)-2\left|V\right|=0.
\end{equation} 

\subsection{Computational Model of Recurrent Gate-Model Quantum Neural Networks}
\begin{theorem} 
The computational model of an ${\rm RQNN}_{QG} $ is a ${\rm {\mathcal D}}\left({\rm RQNN}_{QG} \right)$ diffusion machine with linear transition functions $f_{T} \left({\rm RQNN}_{QG} \right)$.
\end{theorem}
\begin{proof}
Let ${\rm {\mathcal C}}\left({\rm RQNN}_{QG} \right)$ be the constraint machine of ${\rm RQNN}_{QG} $ with linear transition function $f_{T} \left({\rm RQNN}_{QG} \right)=f_{T} \left(\gamma _{\Gamma \left(v\right)} ,x_{v} \right)$. Using the ${\rm {\mathcal G}}_{{\rm RQNN}_{QG} } $ environmental graph, let $\Lambda _{v} $ be a constraint on $f_{T} \left({\rm RQNN}_{QG} \right)$ of ${\rm RQNN}_{QG} $, $v\in V$ as 
\begin{equation} \label{ZEqnNum411472} 
\Lambda _{v} :{\left| \gamma _{v}  \right\rangle} -f_{T} \left({\rm RQNN}_{QG} \right)=0, 
\end{equation} 
where $\left| \gamma _{v}  \right\rangle $ is the quantum state associated with $\gamma _{v}$ state variable of a given $v$ of ${\rm RQNN}_{QG} $.
With respect to the output function  $F_{O} \left({\rm RQNN}_{QG} \right)=F_{O} \left(\gamma _{v} ,x_{v} \right)$ of ${\rm RQNN}_{QG} $ , let $\omega _{v} $ be a constraint on $F_{O} \left({\rm RQNN}_{QG} \right)$ of ${\rm RQNN}_{QG} $, as
\begin{equation} \label{ZEqnNum937370} 
\omega _{v} :\Omega _{v} \circ F_{O} \left({\rm RQNN}_{QG} \right)=0,
\end{equation} 
where $\Omega _{v}$ is another constraint as $\Omega _{v}(F_{O} \left({\rm RQNN}_{QG} \right))=0$.

Since ${\rm RQNN}_{QG} $ is a recurrent network, for all $v\in V$ of ${\rm {\mathcal G}}_{{\rm RQNN}_{QG} } $, a diffuse constraint ${\mathchar'26\mkern-10mu\lambda} \left(Q\left(x\right)\right)$ can be defined via constraints \eqref{ZEqnNum411472} and \eqref{ZEqnNum937370}, as
\begin{equation} \label{ZEqnNum914056} 
\begin{split}
{\mathchar'26\mkern-10mu\lambda} \left(Q\left(x\right)\right)\\
&=\sum _{v\in V}\left(\Lambda _{v} +\omega _{v} \right)-2\left|V\right|\\
&=\sum _{v\in V}\left((\left| \gamma _{v}  \right\rangle-(f_{T} \left({\rm RQNN}_{QG} \right)) +\Omega _{v}(F_{O} \left({\rm RQNN}_{QG} \right)) \right)-2\left|V\right|\\
&=0, 
\end{split}
\end{equation} 
where $x=\left(x_{1} ,\ldots ,x_{\left|V\right|} \right)$, and $Q\left(x\right)=\left(Q\left(x_{1} \right),\ldots ,Q\left(x_{\left|V\right|} \right)\right)$ is a function that maps all vertexes of ${\rm {\mathcal G}}_{{\rm RQNN}_{QG} } $. Therefore, in the presence of \eqref{ZEqnNum914056}, the relation 
\begin{equation} \label{46)} 
{\rm {\mathcal C}}\left({\rm RQNN}_{QG} \right)={\rm {\mathcal D}}\left({\rm RQNN}_{QG} \right),                                                 
\end{equation} 
follows for an ${\rm RQNN}_{QG} $, where ${\rm {\mathcal D}}\left({\rm RQNN}_{QG} \right)$ is the diffusion machine of ${\rm RQNN}_{QG} $. It is because a constraint machine ${\rm {\mathcal C}}\left({\rm RQNN}_{QG} \right)$ that satisfies \eqref{ZEqnNum914056} is, in fact, a diffusion machine ${\rm {\mathcal D}}\left({\rm RQNN}_{QG} \right)$, see also \eqref{xeq}.

In \eqref{ZEqnNum411472}, the $f_{T} \left({\rm RQNN}_{QG} \right)$ state transition function can be defined for a ${\rm {\mathcal D}}\left({\rm RQNN}_{QG} \right)$ via constraint  \eqref{ZEqnNum411472} as 
\begin{equation} \label{ZEqnNum471522} 
f_{T} \left({\rm RQNN}_{QG} \right)={\left| \gamma _{v}  \right\rangle}.
\end{equation} 

Then, let $H_{t} $ be a unit vector for a unitary $U_{t} \left(\theta _{t} \right)$, $t=1,\ldots ,L-1$, defined as 
\begin{equation} \label{ZEqnNum928379} 
H_{t} =x_{t} +iy_{t} ,                                                                  
\end{equation} 
where $x_{t} $ and $y_{t} $ are real values.

Then, let $Z_{t+1} $ be defined via $U(\vec{\theta })$ and \eqref{ZEqnNum928379} as
\begin{equation} \label{ZEqnNum577141} 
Z_{t+1} =U(\vec{\theta })H_{t} +Ex_{t+1} ,                                                     
\end{equation} 
where $E$ is a basis vector matrix \cite{ref25}.

Then, by rewriting $U(\vec{\theta })$ as 
\begin{equation} \label{50)} 
U(\vec{\theta })=\phi +i\varphi ,                                                               
\end{equation} 
where $\phi ,\varphi $ are real parameters, allows us to evaluate $U(\vec{\theta })H_{t} $ as
\begin{equation} \label{51)} 
\left( \begin{matrix}
   \operatorname{Re}( U( {\vec{\theta }}){{H}_{t}})\\
   \operatorname{Im}( U( {\vec{\theta }}){{H}_{t}}\\
\end{matrix} \right)=\left( \begin{matrix}
   \phi  & -\varphi   \\
   \varphi  & \phi   \\
\end{matrix} \right)\left( \begin{matrix}
   {{x}_{t}}  \\
   {{y}_{t}}  \\
\end{matrix} \right)
\end{equation} 
with
\begin{equation} \label{52)} 
H_{t+1} =f_{\sigma }^{{\rm RQNN}_{QG} } \left(Z_{t+1} \right),                                                                  
\end{equation} 
where $H_{t+1}$ is normalized at unity, and function $f_{\sigma }^{{\rm RQNN}_{QG} } \left(\cdot \right)$ is defined as
\begin{equation} \label{ZEqnNum368075} 
f_{\sigma }^{{\rm RQNN}_{QG} } \left(Z\right)=\left\{\begin{array}{l} {Z,{\rm \; if\; }\left|Z\right|_{1} \ge 0{\rm \; }} \\ {0,{\rm \; if\; }\left|Z\right|_{1} <0{\rm \; }} \end{array}\right. , 
\end{equation} 
where ${{\left| \cdot  \right|}_{1}}$ is the $L1$-norm.

Since the ${\rm RQNN}_{QG} $ has linear transition function, \eqref{ZEqnNum368075} is also linear, and allows us to rewrite \eqref{ZEqnNum368075} via the environmental graph representation for a particular $\left(\gamma _{\Gamma \left(v\right)} ,x_{v} \right)$, as 
\begin{equation} \label{54)} 
f_{\sigma }^{{\rm RQNN}_{QG} } \left(Z\right)=\left\{\begin{array}{l} {f_{T} \left(\gamma _{\Gamma \left(v\right)} ,x_{v} \right),{\rm \; if\; }\left|Z\right|_{1} \ge 0{\rm \; }} \\ {0,{\rm \; if\; }\left|Z\right|_{1} <0{\rm \; }} \end{array}\right. , 
\end{equation} 
where $f_{T} \left(\gamma _{\Gamma \left(v\right)} ,x_{v} \right)$ is given in \eqref{ZEqnNum471522}. 

Thus, by setting $t=\nu $, the term $H_{t} $ can be rewritten via \eqref{ZEqnNum471522} and \eqref{ZEqnNum577141} as
\begin{equation} \label{55)} 
H_{t} =x_{\nu } =\operatorname{Re}\left(x_{\nu } \right)+i\operatorname{Im}\left(x_{\nu } \right).                                           
\end{equation} 
Then, the $Y_{t} \left({\rm RQNN}_{QG} \right)$ output of ${\rm RQNN}_{QG} $ is evaluated as
\begin{equation} \label{56)} 
{{Y}_{t}}\left( \text{RQN}{{\text{N}}_{QG}} \right)=W\left( \begin{matrix}
   \operatorname{Re}\left( {{x}_{\nu }} \right)  \\
   \operatorname{Im}\left( {{x}_{\nu }} \right)  \\
\end{matrix} \right),
\end{equation} 
where $W$ is an output matrix \cite{ref25}. 

Then let $\left|\Gamma \left(v\right)\right|=L$, therefore at a particular objective function $f\left(\theta \right)$ of the ${\rm RQNN}_{QG} $, the derivative ${\textstyle\frac{df\left(\theta \right)}{dx_{\nu } }} $ can be evaluated as
\begin{equation} \label{57)} 
\begin{split}
 \frac{df\left( \theta  \right)}{d{{x}_{\nu }}}&=\frac{df\left( \theta  \right)}{d{{x}_{\left| \Gamma \left( v \right) \right|}}}\frac{d{{x}_{\left| \Gamma \left( v \right) \right|}}}{d{{x}_{\nu }}} \\ 
 & =\frac{df\left( \theta  \right)}{d{{x}_{\left| \Gamma \left( v \right) \right|}}}\prod\limits_{k=\nu }^{\left| \Gamma \left( v \right) \right|-1}{\frac{d{{x}_{k+1}}}{d{{x}_{k}}}} \\ 
 & =\frac{df\left( \theta  \right)}{d{{x}_{\left| \Gamma \left( v \right) \right|}}}\prod\limits_{k=\nu }^{\left| \Gamma \left( v \right) \right|-1}{{{D}_{k+1}}{{\left( U( {\vec{\theta }}) \right)}^{T}}},  
\end{split}
\end{equation} 
where 
\begin{equation} \label{58)} 
D_{k+1} =diag\left(Z_{k+1} \right) 
\end{equation} 
is a Jacobian matrix \cite{ref25}. For the norms the relation
\begin{equation} \label{59)} 
\left\| \frac{df\left(\theta \right)}{dx_{\nu } } \right\| \le \left\| \frac{df\left(\theta \right)}{dx_{\left|\Gamma \left(v\right)\right|} } \right\| \prod _{k=\nu }^{\left|\Gamma \left(v\right)\right|-1}\left\| D_{k+1} \left(U(\vec{\theta })\right)^{T} \right\|  , 
\end{equation} 
holds, where 
\begin{equation} \label{60)} 
\left\| D_{k+1} \left(U(\vec{\theta })\right)^{T} \right\| =\left\| D_{k+1} \right\| .                                                          
\end{equation} 
The proof is concluded here.
\end{proof}

\section{Optimal Learning}
\label{sec5}
\subsection{Gate-Model Quantum Neural Network}
\begin{theorem}
A supervised learning is an optimal learning for a ${\rm {\mathcal C}}\left({\rm QNN}_{QG} \right)$.  
\end{theorem}
\begin{proof}
Let $\pi _{v} $ be the compact constraint on $f_{T} \left({\rm QNN}_{QG} \right)$ and $F_{O} \left({\rm QNN}_{QG} \right)$ of ${\rm {\mathcal C}}\left({\rm QNN}_{QG} \right)$ from \eqref{ZEqnNum734924}, and let $A$ be a constraint matrix. Then, \eqref{ZEqnNum734924} can be reformulated as
\begin{equation} \label{ZEqnNum507571} 
\pi _{v} \left(f_{T} \left({\rm QNN}_{QG} \right),F_{O} \left({\rm QNN}_{QG} \right)\right)=Af^{*} \left(x\right)-b\left(x\right)=0.                           
\end{equation} 
where $b\left(x\right)$ is a smooth vector-valued function with compact support \cite{ref23}, ${{f}^{*}}:\mathcal{I}\to {{\mathbb{C}}^{n}}$,
\begin{equation} \label{62)} 
f^{*}(x) =\left(f_{T} \left({\rm QNN}_{QG} \right),F_{O} \left({\rm QNN}_{QG} \right), x\right) 
\end{equation}
is the compact function subject to be determined such that
\begin{equation}\label{eqrel}
\forall x\in \mathcal{X}:{{\pi }_{v}}\left( v,{{f}^{*}}\left( x \right) \right)=0.
\end{equation}
The problem formulated via \eqref{ZEqnNum507571} can be rewritten as
\begin{equation} \label{63)} 
Af^{*} \left(x\right)=b\left(x\right).                                                                    
\end{equation} 
As follows, learning of functions $f_{T} \left({\rm QNN}_{QG} \right)$ and $F_{O} \left({\rm QNN}_{QG} \right)$ of ${\rm {\mathcal C}}\left({\rm QNN}_{QG} \right)$ can be reduced to the determination of function $f^{*} \left(x\right)$, which problem is solvable via the Euler-Lagrange equations \cite{ref23,ref39,ref40}.

Then, let ${{\mathcal{S}}_{L\left( \text{QNN} \right)}}$ be a non-empty supervised learning set defined as a collection 
\begin{equation} \label{64)} 
{{{\mathcal{S}}_{L\left( \text{QNN} \right)}}}:\left\{x_{\kappa } ,y_{\kappa } ,\kappa \in {\rm {\mathbb{N}}}_{\left| \mathcal{X} \right|} \right\},                                                       
\end{equation} 
where $\left(x_{\kappa } ,y_{\kappa } \right)$, $y_{\kappa } =f^{*} \left(x_{\kappa } \right)$ is a supervised pair, and $\left| \mathcal{X} \right|$ is the cardinality of the perceptive space $\mathcal{X}$ associated with ${{{\mathcal{S}}_{L\left( \text{QNN} \right)}}}$.

Since ${{{\mathcal{S}}_{L\left( \text{QNN} \right)}}}$ is non-empty set, $f^{*} \left(x\right)$ can be evaluated by the Euler-Lagrange equations \cite{ref23,ref39,ref40}, as
\begin{equation} \label{ZEqnNum900941} 
f^{*} \left(x\right)=\frac{1}{\ell } \left(-{A}^{T}\lambda \left(x\right)-\frac{1}{\left| \mathcal{X} \right|} \sum _{\kappa =1}^{\left| \mathcal{X} \right|}\left(f^{*} \left(x\right)-y_{\kappa } \right)\Upsilon \left(x-x_{\kappa } \right) \right), 
\end{equation} 
where ${{A}^{T}}$ is the transpose of the constraint matrix $A$, and $\ell $ is a differential operator as 
\begin{equation} \label{67)} 
\ell =\sum _{\kappa =0}^{k}\left(-1\right)^{\kappa }  c_{\kappa } \nabla ^{2\kappa } ,                                                               
\end{equation} 
where $c_{\kappa } $-s are constants, ${{\nabla }^{2}}$ is a Laplacian operator such that ${{\nabla }^{2}}f\left( x \right)=\sum\nolimits_{i}{\partial _{i}^{2}{{f}}\left( x \right)}$; while $\Upsilon $ is as
\begin{equation} \label{68)} 
\Upsilon \left(x-x_{\kappa } \right)=\ell {\rm {\mathcal G}}\left(x,x_{\kappa } \right),                                                             
\end{equation} 
where ${\rm {\mathcal G}}\left(\cdot \right)$ is the Green function of differential operator $\ell $. Since function ${\rm {\mathcal G}}\left(\cdot \right)$ is translation invariant, the relation  
\begin{equation} \label{ZEqnNum591882} 
{\rm {\mathcal G}}\left(x,x_{\kappa } \right)={\rm {\mathcal G}}\left(x-x_{\kappa } \right) 
\end{equation} 
follows. 
Since the constraint that has to be satisfied over the perceptual space $\mathcal{X}$ is given in \eqref{eqrel}, the $\mathcal{L}$ Lagrangian can be defined as
\begin{equation}
\mathcal{L}=\left\langle P{{f}^{*}},P{{f}^{*}} \right\rangle +\int\limits_{{\rm {\mathcal X}}}{\lambda \left( x \right){{\pi }_{v}}\left( x,{{f}^{*}}\left( x \right) \right)dx},
\end{equation}
where $\left\langle \cdot ,\cdot  \right\rangle $ is the inner product operator, while $P$ is defined via \eqref{67)} as
\begin{equation}
\ell ={{P}^{\dagger }}P,
\end{equation}
where ${{P}^{\dagger }}$ is the adjoint of $P$, while $\lambda \left(x\right)$ is the Lagrange multiplier as
\begin{equation} \label{ZEqnNum815037} 
\lambda \left(x\right)=-\left(A{A}^{T}\right)^{-1} \left(\gamma \ell b\left(x\right)+\frac{1}{\left| \mathcal{X} \right|} \sum _{\kappa =1}^{\left| \mathcal{X} \right|}A\left(f^{*} \left(x\right)-y_{\kappa } \right)\Upsilon \left(x-x_{\kappa } \right) \right),          
\end{equation} 
where 
\begin{equation} \label{71)} 
\gamma =\int\limits_{{\rm {\mathcal X}}}{\rm {\mathcal G}}\left(x-x_{\kappa } \right) ,                                                                    
\end{equation} 
and $\ell b$ is as
\begin{equation} \label{ZEqnNum212449} 
\ell b\left(x\right)=-A\left({A}^{T}\lambda \left(x\right)+\frac{1}{\left| \mathcal{X} \right|} \sum _{\kappa =1}^{\left| \mathcal{X} \right|}\left(f^{*} \left(x\right)-y_{\kappa } \right)\Upsilon \left(x-x_{\kappa } \right) \right).                      
\end{equation} 
Then, \eqref{ZEqnNum900941} can be rewritten using \eqref{ZEqnNum815037} and \eqref{ZEqnNum212449} as
\begin{equation} \label{73)} 
f^{*} \left(x\right)=\frac{1}{\gamma \ell } \left(H\left(x\right)+\frac{1}{\left| \mathcal{X} \right|} {\sum _{\kappa =1}^{\left| \mathcal{X} \right|}}\Phi \left(y_{\kappa } -f^{*} \left(x\right)\right)\Upsilon \left(x-x_{\kappa } \right) \right),                       
\end{equation} 
where $H\left(x\right)$ is as 
\begin{equation} \label{74)} 
H\left(x\right)=\gamma {A}^{T}\left(A{A}^{T}\right)^{-1} \ell b\left(x\right) 
\end{equation} 
and $\Phi $ is as
\begin{equation} \label{75)} 
\Phi ={{\mathbf{I}}_{n}} -{A}^{T}\left(A{A}^{T}\right)^{-1} A,                                                          
\end{equation} 
where ${{\mathbf{I}}_{n}}$ is an identity matrix.

Therefore, after some calculations, $f^{*} \left(x\right)$ can be expressed as
\begin{equation} \label{ZEqnNum763496}  
f^{*} \left(x\right)=\frac{1}{\gamma } \int\limits_{{\rm {\mathcal X}}}{\rm {\mathcal G}}\left(z\right)H\left(x-z\right)dz +\sum _{\kappa =1}^{\left| \mathcal{X} \right|}\Phi \chi _{\kappa } {\rm {\mathcal G}}\left(x-x_{\kappa } \right) , 
\end{equation} 
where $\chi _{\kappa } $ is as 
\begin{equation} \label{77)} 
\chi _{\kappa } =\frac{1}{\left| \mathcal{X} \right|} \frac{y_{\kappa } -f^{*} \left(x_{\kappa } \right)}{\gamma } .                                                        
\end{equation} 
The compact constraint of ${\rm {\mathcal C}}\left({\rm QNN}_{QG} \right)$ determined via \eqref{ZEqnNum763496} is optimal, since \eqref{ZEqnNum763496} is the optimal solution of the Euler--Lagrange equations. 

The proof is concluded here.
\end{proof}

\begin{lemma}
There exists a supervised learning for a ${\rm {\mathcal C}}\left({\rm QNN}_{QG} \right)$ with  complexity $\mathcal{O}\left( \left| S \right| \right)$, where $|S|$ is the number arcs (number of gate parameters) of ${\rm {\mathcal G}}_{{\rm QNN}_{QG} }$.
\end{lemma}
\begin{proof}
Let ${\rm {\mathcal G}}_{{\rm QNN}_{QG} } $ be the environmental graph of ${\rm QNN}_{QG} $, such that ${\rm QNN}_{QG} $ is characterized via $\vec{\theta }$ (see \eqref{ZEqnNum837426}). 

The optimal supervised learning method of a ${\rm {\mathcal C}}\left({\rm QNN}_{QG} \right)$ is derived through the utilization of the ${\rm {\mathcal G}}_{{\rm QNN}_{QG} } $ environmental graph of ${\rm QNN}_{QG} $, as follows. 

The ${\rm {\mathcal A}}_{{\rm {\mathcal C}}\left({\rm QNN}_{QG} \right)} $ learning process of ${\rm {\mathcal C}}\left({\rm QNN}_{QG} \right)$ in the ${\rm {\mathcal G}}_{{\rm QNN}_{QG} } $ structure is given in Algorithm 1.

\setcounter{algocf}{0}
\begin{algo}
\small
  \DontPrintSemicolon
\caption{Supervised learning for a ${\rm {\mathcal C}}\left({\rm QNN}_{QG} \right)$}

\textbf{Step 1}. (Initialization.) Set $Y_{n+1}^{\left(r\right)} $, $r=1,\ldots ,R$, where $R$ is the number of total measurements applied for ${\rm QNN}_{QG} $, and let $U(\vec{\theta })$ as given in \eqref{ZEqnNum335090}.

\textbf{Step 2}. (Quantum evolution phase and parameter set.) Apply the unitary sequence $U(\vec{\theta })$ of ${\rm QNN}_{QG} $ realized via the $L$ unitaries to produce output ${\left| Y \right\rangle} $. 
Let us assume that the node set $V$ of ${\rm {\mathcal G}}_{{\rm QNN}_{QG} } $ is topologically sorted via a topological ordering function $f_{\angle } ({\rm {\mathcal G}}_{{\rm QNN}_{QG} })$. 

Then, for a given $x_{U_{i} \left(\theta _{i} \right)} $ (\eqref{23)}) set
\begin{equation} \label{wi} 
V_{U_{i} \left(\theta _{i} \right)} =v_{U_{i}}+Q_{U_{i} \left(\theta _{i} \right)}, 
\end{equation} 
where
\begin{equation} \label{qi} 
Q_{U_{i} \left(\theta _{i} \right)} =\sum _{h\in \Xi \left(i\right)}\theta _{hi}V_{U_{h} \left(\theta _{h} \right)} +{{B}_{{{v}_{{{U}_{i}}}}}}, 
\end{equation} 
where $\Xi \left(i\right)$ refers to the parent set of $v_{U_{i}} $, ${{B}_{{{v}_{{{U}_{i}}}}}}$ is a bias relative to $v_{U_{i}}$.

\textbf{Step 3}. (Error initialization.) Set the ${\rm P}^{(r)} ({\rm {\mathcal G}}_{{\rm QNN}_{QG} })$ post-processing associated to the $r$-th measurement round on ${\rm {\mathcal G}}_{{\rm QNN}_{QG} } $, as follows. 

For $i=1,\ldots ,L$, set $W_{U_{i} \left(\theta _{i} \right)}$ as
\begin{equation} \label{c} 
W_{U_{i} \left(\theta _{i} \right)} =\sum _{h\in \Xi \left(i\right)}\theta _{hi} V_{U_{h} \left(\theta _{h} \right)}.
\end{equation} 

For $i=1,\ldots ,L-1$, compute the error $\delta _{U_{i} \left(\theta _{i} \right)} $ associated to $U_{i} \left(\theta _{i} \right)$ as
\begin{equation} \label{ZEqnNum919206} 
\delta _{U_{i} \left(\theta _{i} \right)} =\frac{dW_{U_{L} \left(\theta _{L} \right)} }{dQ_{U_{i} \left(\theta _{i} \right)} }.               
\end{equation} 

For $i=L$, set 
\begin{equation} \label{85)} 
\delta _{U_{L} \left(\theta _{L} \right)} =\frac{d{\rm {\mathcal L}}(x_{0} ,\tilde{l}(z))}{dQ_{U_{L} \left(\theta _{L} \right)} }.                                                
\end{equation}

\textbf{Step 4}. (Gate parameter updating.) Set a gate parameter modification vector $\vec{\Delta }\theta $ with an $i$-th element $\vec{\Delta }\theta _{i} $ as 
\begin{equation} \label{87)} 
\vec{\Delta }\theta _{i} =W_{U_{i} \left(\theta _{i} \right)},
\end{equation} 
where $W_{U_{i} \left(\theta _{i} \right)}$ is given in \eqref{c}. Update the gate parameters in a backpropagated manner from unitary ${{U}_{L}}\left( {{\theta }_{L}} \right)$ to ${{U}_{1}}\left( {{\theta }_{1}} \right)$, as follows.

For $z={L,\ldots ,1}$:
\begin{quote}
If $\vec{\Delta }{{\theta }_{z}}= 1$, update $\theta _{z}$ as
\begin{equation} \label{ZEqnNum414380} 
\theta '_{z} =\theta _{z}.
\end{equation} 

If $\vec{\Delta }{{\theta }_{z}}\ne 1$, update $\theta _{z}$ as
\begin{equation} \label{ZEqnNum426724} 
\theta ' _{z} =(\vec{\Delta }\theta _{z}) \theta _{z}.
\end{equation} 
\end{quote}

\textbf{Step 5}. (Gradient computation). For ${i=2,\ldots ,L}$, and $j\in \Xi \left(i\right)$, determine the gradient $g_{U_{i} \left(\theta _{i} \right),U_{j} \left(\theta _{j} \right)} $ between unitaries $U_{i} \left(\theta _{i} \right)$ and $U_{j} \left(\theta _{j} \right)$, as 
\begin{equation} \label{90)} 
g_{U_{i} \left(\theta _{i} \right),U_{j} \left(\theta _{j} \right)} =\delta '_{U_{i} \left(\theta _{i} \right)} W_{U_{j} \left(\theta _{j} \right)},                                                       
\end{equation} 
where $\delta '_{U_{i} \left(\theta _{i} \right)}$ is the updated error evaluated as
\begin{equation}\label{err}
\delta '_{U_{i} \left(\theta _{i} \right)}=(\vec{\Delta }\theta _{i})\delta _{U_{i} \left(\theta _{i} \right)}.
\end{equation}

\textbf{Step 6}. Apply steps 1-5, for $\forall r$ measurements.
 
\end{algo}

%\setcounter{algocf}{0}
%\begin{algo}
%  \DontPrintSemicolon
%\caption{\textit{Optimal supervised learning for a ${\rm {\mathcal C}}\left({\rm QNN}_{QG} \right)$. (cont.)}}
%
%\end{algo}

The optimality of Algorithm 1 arises from the fact that in Step 4, the gradient computation involves all the gate parameters of the ${\rm QNN}_{QG}$, and the gate parameter updating procedure has a computational complexity $\mathcal{O}\left( \left| S \right| \right)$. The ${\rm QNN}_{QG}$ complexity is yielded from the gate parameter updating mechanism that utilizes backpropagated classical side information for the learning method.

The proof is concluded here.
\end{proof}

\subsubsection{Description and Method Validation}

The detailed steps and validation of Algorithm 1 are as follows.

In Step 1, the number $R$ of measurement rounds is set. 

Step 2 is the quantum evolution phase of ${\rm QNN}_{QG} $ that yields an output quantum system ${\left| Y \right\rangle} $ via forward propagation of quantum information through the unitary sequence $U(\vec{\theta })$ realized via the $L$ unitaries. Then, a parameterization follows for each $x_{U_{i} \left(\theta _{i} \right)} $, and the terms $W_{U_{i} \left(\theta _{i} \right)} $ and $Q_{U_{i} \left(\theta _{i} \right)} $ are defined to characterize the $\theta _{i} $ angles of the $U_{i} \left(\theta _{i} \right)$ unitary operations in the ${\rm QNN}_{QG} $. 

In Step 3, side information initializations are made for the error computations. A given $W_{U_{i} \left(\theta _{i} \right)} $ is set as a cumulative quantity with respect to the parent set $\Xi \in i$ of unitary $U_{i} \left(\theta _{i} \right)$ in ${\rm QNN}_{QG} $. 

Note, that \eqref{qi} and \eqref{c} represent side information, thus the gate parameter $\theta_{hi}$ is used to identify a particular unitary $U(\theta_{hi})$.

Let ${{{\mathcal{G}}'_{\text{QN}{{\text{N}}_{QG}}}}}$ be the the environmental graph of ${\rm QNN}_{QG} $ such that the directions of quantum links are reversed. It can be verified that for a ${{{\mathcal{G}}'_{\text{QN}{{\text{N}}_{QG}}}}}$, $\delta _{U_{i} \left(\theta _{i} \right)}$ from \eqref{ZEqnNum919206} can be rewritten as
\begin{equation} \label{r1}
\delta _{U_{i} \left(\theta _{i} \right)} =\sum _{h\in \Xi \left(i\right)}\frac{dW_{U_{L} \left(\theta _{L} \right)} }{dQ_{U_{h} \left(\theta _{h} \right)} }  \frac{dQ_{U_{h} \left(\theta _{h} \right)} }{dW_{U_{i} \left(\theta _{i} \right)} } \frac{dW_{U_{i} \left(\theta _{i} \right)} }{dQ_{U_{i} \left(\theta _{i} \right)} }=Q_{U_{i} \left(\theta _{i} \right)} \sum _{h\in \Xi \left(i\right)}\theta _{hi} \delta _{U_{h} \left(\theta _{h} \right)}, 
\end{equation}
and $\delta _{U_{L} \left(\theta _{L} \right)}$ can be evaluated as given in \eqref{85)} 
\begin{equation} \label{ZEqnNum434530} 
\delta _{U_{L} \left(\theta _{L} \right)} =\frac{d{\rm {\mathcal L}}(x_{0} ,\tilde{l}(z))}{dQ_{U_{L} \left(\theta _{L} \right)} }, 
\end{equation}
while the term $\delta _{U_{i} \left(\theta _{i} \right)} W_{U_{j} \left(\theta _{j} \right)} $  for each $U_{i} \left(\theta _{i} \right)$ can be rewritten as 
\begin{equation} \label{86)} 
\delta _{U_{i} \left(\theta _{i} \right)} W_{U_{j} \left(\theta _{j} \right)} =\frac{d{\rm {\mathcal L}}(x_{0} ,\tilde{l}(z))}{d\theta _{ij} } =\frac{d{\rm {\mathcal L}}(x_{0} ,\tilde{l}(z))}{dQ_{U_{i} \left(\theta _{i} \right)} } \frac{dQ_{U_{i} \left(\theta _{i} \right)} }{d\theta _{ij} }.  
\end{equation} 

Since \eqref{ZEqnNum426724} and \eqref{ZEqnNum414380} are defined via the non-reversed  ${\rm {\mathcal G}}_{{\rm QNN}_{QG} } $, for a given unitary the $\Gamma $ children set is used. The utilization of the $\Xi $ parent set with reversed link directions in ${{{\mathcal{G}}'_{\text{QN}{{\text{N}}_{QG}}}}}$ (see \eqref{r1}, \eqref{ZEqnNum434530}, \eqref{86)}) is therefore analogous to the use of the $\Gamma $ children set with non-reversed link directions in  ${\rm {\mathcal G}}_{{\rm QNN}_{QG} } $. It is because classical side information is available in arbitrary directions in ${\rm {\mathcal G}}_{{\rm QNN}_{QG} } $. 

First, we consider the situation, if $i=1,\ldots ,L-1$, thus the error calculations are associated to unitaries $U_{1} \left(\theta _{1} \right),\ldots ,U_{L-1} \left(\theta _{L-1} \right)$, while the output unitary $U_{L} \left(\theta _{L} \right)$ is proposed for the $i=L$ case. 

In ${\rm {\mathcal G}}_{{\rm QNN}_{QG} } $, the error quantity $\delta _{U_{i} \left(\theta _{i} \right)} $ associated to $U_{i} \left(\theta _{i} \right)$ is determined, where $W_{U_{L} \left(\theta _{L} \right)} $ is associated to the output unitary $U_{L} \left(\theta _{L} \right)$. Only forward steps are required to yield $W_{U_{L} \left(\theta _{L} \right)} $ and $Q_{U_{L} \left(\theta _{L} \right)} $. Then, utilizing the chain rule and using the children set $\Gamma \left(i\right)$ of a particular unitary $U_{i} \left(\theta _{i} \right)$, the term ${dW_{U_{L} \left(\theta _{L} \right)}  \mathord{\left/{\vphantom{dW_{U_{L} \left(\theta _{L} \right)}  dQ_{U_{i} \left(\theta _{i} \right)} }}\right.\kern-\nulldelimiterspace} dQ_{U_{i} \left(\theta _{i} \right)} } $ in $\delta _{U_{i} \left(\theta _{i} \right)} $ can be rewritten as ${\textstyle\frac{dW_{U_{L} \left(\theta _{L} \right)} }{dQ_{U_{i} \left(\theta _{i} \right)} }} =\sum _{h\in \Gamma \left(i\right)}{\textstyle\frac{dW_{U_{L} \left(\theta _{L} \right)} }{dQ_{U_{h} \left(\theta _{h} \right)} }}  {\textstyle\frac{dQ_{U_{h} \left(\theta _{h} \right)} }{dW_{U_{i} \left(\theta _{i} \right)} }} {\textstyle\frac{dW_{U_{i} \left(\theta _{i} \right)} }{dQ_{U_{i} \left(\theta _{i} \right)} }} $. In fact, this term equals to $Q_{U_{i} \left(\theta _{i} \right)} \sum _{h\in \Gamma \left(i\right)}\theta _{hi}  \delta _{U_{h} \left(\theta _{h} \right)} $, where $\delta _{U_{h} \left(\theta _{h} \right)} $ is the error associated to a $U_{h} \left(\theta _{h} \right)$, such that $U_{h} \left(\theta _{h} \right)$ is a children unitary of $U_{i} \left(\theta _{i} \right)$. The $\delta _{U_{h} \left(\theta _{h} \right)} $ error quantity associated to a children unitary $U_{h} \left(\theta _{h} \right)$ of $U_{i} \left(\theta _{i} \right)$ can also be determined in the same manner, that yields $\delta _{U_{h} \left(\theta _{h} \right)} ={dW_{U_{L} \left(\theta _{L} \right)}  \mathord{\left/{\vphantom{dW_{U_{L} \left(\theta _{L} \right)}  dQ_{U_{h} \left(\theta _{h} \right)} }}\right.\kern-\nulldelimiterspace} dQ_{U_{h} \left(\theta _{h} \right)} } $. As follows, by utilizing side information in ${\rm {\mathcal G}}_{{\rm QNN}_{QG} } $ allows us to determine $\delta _{U_{i} \left(\theta _{i} \right)} $ via the ${\rm {\mathcal L}}\left(\cdot \right)$ loss function and the $\Gamma \left(i\right)$ children set of unitary $U_{i} \left(\theta _{i} \right)$, that yields the quantity given in \eqref{ZEqnNum919206}. 

The situation differs if the error computations are made with respect to the output system, thus for the $L$-th unitary $U_{L} \left(\theta _{L} \right)$. In this case, the utilization of the loss function ${\rm {\mathcal L}}(x_{0} ,\tilde{l}\left(z\right))$ allows us to use the simplified formula of $\delta _{U_{L} \left(\theta _{L} \right)} ={d{\rm {\mathcal L}}(x_{0} ,\tilde{l}\left(z\right)) \mathord{\left/{\vphantom{d{\rm {\mathcal L}}(x_{0} ,\tilde{l}\left(z\right)) dQ_{U_{L} \left(\theta _{L} \right)} }}\right.\kern-\nulldelimiterspace} dQ_{U_{L} \left(\theta _{L} \right)} } $, as given in \eqref{85)}. Taking the ${\textstyle\frac{d{\rm {\mathcal L}}(x_{0} ,\tilde{l}\left(z\right))}{d\theta _{ij} }} $ derivative of the loss function ${\rm {\mathcal L}}(x_{0} ,\tilde{l}\left(z\right))$ with respect to the angle $\theta _{ij} $ yields ${\textstyle\frac{d{\rm {\mathcal L}}(x_{0} ,\tilde{l}\left(z\right))}{dQ_{U_{i} \left(\theta _{i} \right)} }} {\textstyle\frac{dQ_{U_{i} \left(\theta _{i} \right)} }{d\theta _{ij} }} $, that is, in fact equals to $\delta _{U_{i} \left(\theta _{i} \right)} W_{U_{j} \left(\theta _{j} \right)} $.

In Step 4, the quantities defined in the previous steps are utilized in the ${\rm QNN}_{QG} $ for the error calculations. The errors are evaluated and updated in a backpropagated manner from unitary ${{U}_{L}}\left( {{\theta }_{L}} \right)$ to ${{U}_{1}}\left( {{\theta }_{1}} \right)$. Since it requires only side information these steps can be achieved via a ${\rm P} ({\rm {\mathcal G}}_{{\rm QNN}_{QG} })$ post-processing (along with Step 3). First, a gate parameter modification vector $\vec{\Delta }\theta $ is defined, such that its $i$-th element, $\vec{\Delta }\theta _{i} $, is associated with the modification of the $\theta _{i} $ gate parameter of an $i$-th unitary $U_{i} \left(\theta _{i} \right)$.  

The $i$-th element $\vec{\Delta }\theta _{i} $ is initialized as $\vec{\Delta }\theta _{i} =W_{U_{i} \left(\theta _{i} \right)} $. If $\vec{\Delta }\theta _{i} $ equals to 1, then no modification is required in the $\theta _{i} $ gate parameter of $U_{i} \left(\theta _{i} \right)$. In this case, the $\delta _{U_{i} \left(\theta _{i} \right)} $ error quantity of $U_{i} \left(\theta _{i} \right)$ can be determined via a simple summation, using the children set of $U_{i} \left(\theta _{i} \right)$, as $\delta _{U_{i} \left(\theta _{i} \right)} =\sum _{j\in \Gamma \left(i\right)}\theta '_{ij} \delta _{U_{j} \left(\theta _{j} \right)}  $, where $U_{j} \left(\theta _{j} \right)$ is a children of $U_{i} \left(\theta _{i} \right)$, as it is given in \eqref{ZEqnNum414380}. On the other hand, if $\vec{\Delta }\theta _{i} \ne 1$, then the $\theta _{i} $ gate parameter of $U_{i} \left(\theta _{i} \right)$ requires a modification. In this case, summation $\sum _{j\in \Gamma \left(i\right)}\theta _{ij} \delta _{U_{j} \left(\theta _{j} \right)}  $ has to be weighted by the actual $\vec{\Delta }\theta _{i} $ to yield $\delta _{U_{i} \left(\theta _{i} \right)} $. This situation is obtained in \eqref{ZEqnNum426724}. 

According to the update mechanism of \eqref{87)}-\eqref{ZEqnNum426724}, for $z=L-1,\ldots ,1$, the errors are updated via \eqref{err} as follows.
At $z={L}$ and $\vec{\Delta }{{\theta }_{z}}= 1$, $\delta _{U_{z} \left(\theta _{z} \right)}$ is as
\begin{equation} \label{a414380} 
\delta '_{U_{z} \left(\theta _{z} \right)} ={\delta _{U_{L} \left(\theta _{L} \right)}}.
\end{equation} 
while at $\vec{\Delta }{{\theta }_{z}}\ne 1$, $\delta _{U_{z} \left(\theta _{z} \right)}$ is updated as
\begin{equation} \label{a426724} 
\delta '_{U_{z} \left(\theta _{z} \right)} = (\vec{\Delta }{{\theta }_{z}}) \delta _{U_{L} \left(\theta _{L} \right)}.
\end{equation} 
For $z={L-1,\ldots ,1}$, if $\vec{\Delta }{{\theta }_{z}}= 1$, then $\delta _{U_{z} \left(\theta _{z} \right)}$ is as
\begin{equation} \label{a414380b} 
\delta '_{U_{z} \left(\theta _{z} \right)} =\delta _{U_{z} \left(\theta _{z} \right)}=\sum _{j\in \Gamma \left(z\right)}\theta _{zj} \delta _{U_{j} \left(\theta _{j} \right)}.
\end{equation} 
while, if $\vec{\Delta }{{\theta }_{z}}\ne 1$, then
\begin{equation} \label{a426724b} 
\delta '_{U_{z} \left(\theta _{z} \right)} = (\vec{\Delta }{{\theta }_{z}})  \sum _{j\in \Gamma \left(z\right)}\theta _{zj} \delta _{U_{j} \left(\theta _{j} \right)}= \sum _{j\in \Gamma \left(z\right)}\theta '_{zj} \delta _{U_{j} \left(\theta _{j} \right)}.
\end{equation} 

In Step 5, for a given unitary $U_{i} \left(\theta _{i} \right)$, $i=2,\ldots ,L$  and for its parent $U_{j} \left(\theta _{j} \right)$, the $g_{U_{i} \left(\theta _{i} \right),U_{j} \left(\theta _{j} \right)} $ gradient is computed via the $\delta _{U_{i} \left(\theta _{i} \right)} $ error quantity derived from \eqref{ZEqnNum414380}-\eqref{ZEqnNum426724} for $U_{i} \left(\theta _{i} \right)$, and by the $W_{U_{j} \left(\theta _{j} \right)} $ quantity associated to parent $U_{j} \left(\theta _{j} \right)$. (For $U_{1} \left(\theta _{1} \right)$ the parent set $\Xi \left(1\right)$ is empty, thus $i>1$.)
The computation of $g_{U_{i} \left(\theta _{i} \right),U_{j} \left(\theta _{j} \right)} $ is performed for all $U_{j} \left(\theta _{j} \right)$ parents of $U_{i} \left(\theta _{i} \right)$, thus \eqref{90)} is determined for $\forall j$, $j\in \Xi \left(i\right)$. By the chain rule,
\begin{equation}
\begin{split}
g_{U_{i} \left(\theta _{i} \right),U_{j} \left(\theta _{j} \right)}= \delta '_{U_{i} \left(\theta _{i} \right)} W_{U_{j} \left(\theta _{j} \right)} &= \frac{dW_{U_{L} \left(\theta _{L} \right)} }{d \theta '_{ij} } \\
&= \frac{dW_{U_{L} \left(\theta _{L} \right)} }{dQ_{U_{i} \left(\theta _{i} \right)} }  \frac{dQ_{U_{i} \left(\theta _{i} \right)} }{d \theta '_{ij} }\\
&= \frac{dW_{U_{L} \left(\theta _{L} \right)} }{dQ_{U_{i} \left(\theta _{i} \right)} }  \frac{d(\sum _{h\in \Xi \left(i\right)}\theta _{hi} W _{U_{h} \left(\theta _{h} \right)}) }{d \theta '_{ij} }.
\end{split}
\end{equation}
Since for $i=L$, $\delta _{U_{L} \left(\theta _{L} \right)}$ is as given in \eqref{85)}, the gradient can be rewritten via \eqref{86)} as
\begin{equation}
\begin{split}
g_{U_{i} \left(\theta _{i} \right),U_{j} \left(\theta _{j} \right)}= \frac{d{\rm {\mathcal L}}(x_{0} ,\tilde{l}(z))}{d\theta '_{ij} }.
\end{split}
\end{equation}

Finally, Step 6 utilizes the number $R$ of measurements to extend the results for all measurement rounds, $r=1,\ldots ,R$. Note that in each round a measurement operator is applied, for simplicity it is omitted from the description. 

Since the algorithm requires no reversed quantum links, i.e. ${{{\mathcal{G}}'_{\text{QN}{{\text{N}}_{QG}}}}}$ for the computations of \eqref{ZEqnNum414380}-\eqref{ZEqnNum426724}, the gradient of the loss in \eqref{90)} with respect to the gate parameter can be determined in an optimal way for ${\rm QNN}_{QG}$ networks, by the utilization of side information in ${{{\mathcal{G}}_{\text{QN}{{\text{N}}_{QG}}}}}$.

The steps and quantities of the learning procedure (Algorithm 1) of a ${\rm QNN}_{QG} $ are illustrated in \fref{fig2}. The ${\rm QNN}_{QG} $ network realizes the unitary $U(\vec{\theta })$. The quantum information is propagated through quantum links (solid lines) between the unitaries, while the auxiliary classical information is propagated via classical links in the network (dashed lines). An $i$-th node is represented via unitary $U_{i} \left(\theta _{i} \right)$. 

For an $i$-th unitary, $U_{i} \left(\theta _{i} \right)$, parameters $W_{U_{i} \left(\theta _{i} \right)} $, $Q_{U_{i} \left(\theta _{i} \right)} $ and $\delta _{U_{i} \left(\theta _{i} \right)} ={dW_{U_{L} \left(\theta _{L} \right)}  \mathord{\left/{\vphantom{dW_{U_{L} \left(\theta _{L} \right)}  dQ_{U_{i} \left(\theta _{i} \right)} }}\right.} dQ_{U_{i} \left(\theta _{i} \right)} } $ for $i<L$, are computed, where $W_{U_{L} \left(\theta _{L} \right)} =\sum _{j\in \Xi \left(L\right)}\theta _{Lj} V_{U_{j} \left(\theta _{j} \right)}  $. For the output unitary, $\delta _{U_{L} \left(\theta _{L} \right)} ={d{\rm {\mathcal L}}(x_{0} ,\tilde{l}\left(z\right)) \mathord{/{\vphantom{d{\rm {\mathcal L}}(x_{0} ,\tilde{l}\left(z\right)) dQ_{U_{L} \left(\theta _{L} \right)} }}} dQ_{U_{L} \left(\theta _{L} \right)} } $. 
Parameters $W_{U_{i} \left(\theta _{i} \right)}$ and $Q_{U_{i} \left(\theta _{i} \right)}$ are determined via forward propagation of side information, the $\delta _{U_{i} \left(\theta _{i} \right)}$ quantities are evaluated via backward propagation of side information. Finally, the gradients, $g_{U_{i} \left(\theta _{i} \right),U_{j} \left(\theta _{j} \right)} =\delta _{U_{i} \left(\theta _{i} \right)} W_{U_{j} \left(\theta _{j} \right)}, $ are computed.

 \begin{center}
\begin{figure*}[h!]
%\vspace{-0.4cm}
\begin{center}
\includegraphics[angle = 0,width=1\linewidth]{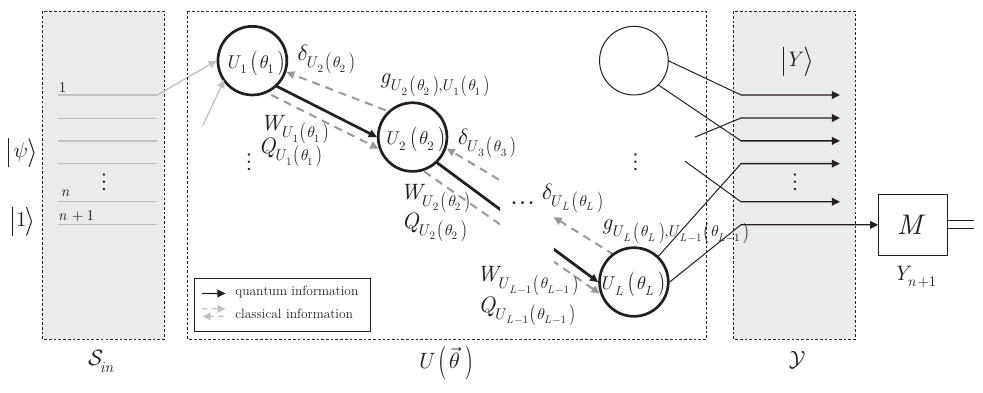}
\caption{The learning method for a ${\rm QNN}_{QG} $. The ${\rm QNN}_{QG} $ network realizes unitary $U(\vec{\theta })$ as a sequence of $L$ unitaries, $U(\vec{\theta })=U_{L} \left(\theta _{L} \right)U_{L-1} \left(\theta _{L-1} \right)\ldots U_{1} \left(\theta _{1} \right)$. The algorithm determines the gradient of the loss with respect to the $\theta $ gate parameter, at a particular loss function ${\rm {\mathcal L}}(x_{0} ,\tilde{l}\left(z\right))$. All quantum information propagates forward via quantum links (solid lines), classical side information can propagate arbitrarily (dashed lines).} 
 \label{fig2}
 \end{center}
\end{figure*}
\end{center}

\subsection{Recurrent Gate-Model Quantum Neural Network}
In classical neural networks, backpropagation \cite{ref24,ref25,ref26} (backward propagation of errors) is a supervised learning method that allows to determine the gradients to learn the weights in the network. In this section, we show that for a recurrent gate-model QNN, a backpropagation method is optimal.
\begin{theorem}
A backpropagation in ${\rm {\mathcal G}}_{{\rm RQNN}_{QG} } $ is an optimal learning in the sense of gradient descent. 
\end{theorem}
\begin{proof}
In an ${\rm RQNN}_{QG} $, the backward classical links provide feedback side information for the forward propagation of quantum information in multiple measurement rounds. The backpropagated side information is analogous to feedback loops, i.e, to recurrent cycles over time. The aim of the learning method is to optimize the gate parameters of the unitaries of the ${\rm RQNN}_{QG} $ quantum network via a supervised learning, using the side information available from the previous $k=1,\ldots ,r-1$ measurement rounds at a particular measurement round $r$. 

Let ${\rm {\mathcal G}}_{{\rm RQNN}_{QG} } $ be the environmental graph of ${\rm RQNN}_{QG} $, and $f_{T} \left({\rm RQNN}_{QG} \right)$ be the transition function of an ${\rm RQNN}_{QG} $. Then the $\gamma _{v} $ constraint is defined via ${\rm {\mathcal G}}_{{\rm RQNN}_{QG} } $ as
\begin{equation} \label{91)} 
{\left| \gamma _{v}  \right\rangle} =f_{T} \left({\rm RQNN}_{QG} \right)=f_{T} \left(\gamma _{\Gamma \left(v\right)} ,x_{v} \right),                                                 
\end{equation} 
while the constraint $\Omega _{v} $ on the output $F\left(\gamma _{v} ,x_{v} \right)$ of ${\rm RQNN}_{QG} $ is defined via $\omega _{v} =0$ as \cite{ref23, ref26,ref27}
\begin{equation} \label{92)} 
\omega _{v} :\Omega _{v} F\left(f_{T} \left({\rm RQNN}_{QG} \right),x_{v} \right)=\Omega _{v} \circ F\left(f_{T} \left(\gamma _{\Gamma \left(v\right)} ,x_{v} \right),x_{v} \right)=0.                    
\end{equation} 
Utilizing the structure of the ${\rm {\mathcal G}}_{{\rm RQNN}_{QG} } $ environmental graph allows us to define a modified version of the backpropagation through time algorithm \cite{ref24} to the ${\rm RQNN}_{QG} $. 

The learning of ${\rm {\mathcal D}}\left({\rm RQNN}_{QG} \right)$ with constraints \eqref{ZEqnNum411472}, \eqref{ZEqnNum937370}, and \eqref{ZEqnNum914056} is given in Algorithm 2, depicted as ${\rm {\mathcal A}}_{{\rm {\mathcal D}}\left({\rm RQNN}_{QG} \right)} $. 

\setcounter{algocf}{1}
\begin{algo}
\small
  \DontPrintSemicolon
\caption{Optimal learning method for a ${\rm {\mathcal D}}\left({\rm RQNN}_{QG} \right)$}

\textbf{Step 1}. (Parameter initialization.) Set the number $R$ of measurement rounds.

For an $r$-th measurement round, $r=1,\ldots ,R$, let 
\begin{equation} \label{ZEqnNum398970} 
{\left| \psi _{r}  \right\rangle} ={| z^{\left(r\right)} \rangle}  
\end{equation} 
be the input quantum system of  ${\rm RQNN}_{QG} $, where $z^{\left(r\right)} $ is an $n$-length string, as 
\begin{equation} \label{ZEqnNum193600} 
z^{\left(r\right)} =z_{r,1} z_{r,2} \ldots z_{r,n} .                                                              
\end{equation}

\textbf{Step 2}. (Quantum evolution phase.) Evaluate the output system ${\left| Y_{r}  \right\rangle} $ of ${\rm RQNN}_{QG} $ as
\begin{equation} \label{ZEqnNum396843} 
{\left| Y_{r}  \right\rangle} =U(\vec{\theta }_{r}){\left| \psi _{r}  \right\rangle} {\left| \varphi _{r}  \right\rangle} =U(\vec{\theta }_{r}){\left| \psi _{r}  \right\rangle} {\left| 1 \right\rangle} =U(\vec{\theta }_{r}){| z^{\left(r\right)} ,1 \rangle} ,                            
\end{equation} 
where $\vec{\theta }_{r} $ is the gate-parameter vector associated to the $L$ unitaries of ${\rm RQNN}_{QG} $ as
\begin{equation} \label{4)} 
\vec{\theta }_{r} =\left(\theta _{r,1} ,\ldots ,\theta _{r,L-1} ,\theta _{r,L} \right)^{T} ,                                                          
\end{equation} 
and an $r$-th unitary sequence is as
\begin{equation} \label{a5)} 
U(\vec{\theta }_{r})=U_{L} \left(\theta _{r,L} \right)U_{L-1} \left(\theta _{r,L-1} \right)\ldots U_{1} \left(\theta _{r,1} \right),                                        
\end{equation} 
where $U_{i} \left(\theta _{r,i} \right)$ is the $i$-th unitary of $U(\vec{\theta }_{r})$.

\textbf{Step 3}. (Post-processing initialization.) For a given $r$, initialize the  ${\rm P}^{\left(r\right)} \left({\rm RQNN}_{QG} \right)$ post-processing as follows. 
Define set ${\rm {\mathcal S}}(\vec{\theta }_{r})$ as
\begin{equation} \label{ZEqnNum522731} 
{\rm {\mathcal S}}(\vec{\theta }_{r})=\{\vec{\theta }_{r} ,B_{r}\},   
\end{equation} 
where $B_{r} $ is a bias.

Define $\Phi _{r} $ as
\begin{equation} \label{ZEqnNum636867} 
\Phi _{r} =z^{\left(r\right)} +U(\vec{\theta }_{r-1})+B_{r} ,                                                          
\end{equation} 
where $U(\vec{\theta }_{r-1})$ is the unitary sequence $U_{L} \left(\theta _{r-1,L} \right)U_{L-1} \left(\theta _{r-1,L-1} \right)\ldots U_{1} \left(\theta _{r-1,1} \right)$, of the $\left(r-1\right)$-th round.

Using \eqref{ZEqnNum636867}, evaluate quantity $\xi _{r,k} $ as
\begin{equation} \label{ZEqnNum809132} 
\xi _{r,k} =\frac{d\Phi _{r} }{d\Phi _{k} } =\prod _{i=k+1}^{r}\frac{d\Phi _{i} }{d\Phi _{i-1} },   
\end{equation} 
where $\Phi _{k} $ belongs to the $k$-th measurement round, $k<r$.

\end{algo}

\setcounter{algocf}{1}
\begin{algo}
\small
  \DontPrintSemicolon
\caption{Optimal learning method for a ${\rm {\mathcal D}}\left({\rm RQNN}_{QG} \right)$ (cont.)}

\textbf{Step 4}. (Gradient computations). Using \eqref{ZEqnNum809132}, compute the $g_{r} $ loss function gradient of the $r$-th round as
\begin{equation} \label{ZEqnNum677383} 
g_{r} =\frac{{\rm {\mathcal L}}(x_{0,r} ,\tilde{l}(z^{\left(r\right)}))}{d{\rm {\mathcal S}}(\vec{\theta }_{r})} =\sum _{k=1}^{r}\frac{{\rm {\mathcal L}}(x_{0,r} ,\tilde{l}(z^{\left(r\right)}))}{d\Phi _{r} }  \xi _{r,k} \frac{d\tilde{\Phi }_{k} }{d{\rm {\mathcal S}}(\vec{\theta }_{r})} ,                                      
\end{equation} 
where ${\rm {\mathcal L}}(x_{0,r} ,\tilde{l}(z^{\left(r\right)}))$ is the loss function of the $r$-th round,
\begin{equation} 
{\rm {\mathcal L}}(x_{0,r} ,\tilde{l}(z^{\left(r\right)}))=1-l\left(z^{\left(r\right)}\right)\tilde{l}\left(z^{\left(r\right)}\right), 
\end{equation}
where $x_{0,r} ={\left| z^{\left(r\right)} ,1 \right\rangle} $ identifies the input system of ${\rm RQNN}_{QG} $ in the $r$-th round, and $\tilde{l}\left(z^{\left(r\right)} \right)$ is as
\begin{equation} \label{ZEqnNum648611} 
\tilde{l}\left(z^{\left(r\right)} \right)={\langle z^{\left(r\right)} ,1 \mathrel{| \vphantom{z^{\left(r\right)} ,1 (U(\vec{\theta }_{r}))^{\dag } Y_{n+1}^{(r)} U(\vec{\theta }_{r})|z^{\left(r\right)} ,1} \kern-\nulldelimiterspace} (U(\vec{\theta }_{r}))^{\dag } Y_{n+1}^{(r)} U(\vec{\theta }_{r})|z^{\left(r\right)} ,1 \rangle} ,                                        
\end{equation} 
where $\tilde{l}\left(z^{\left(r\right)} \right)$ is the predicted value of the binary label $l\left(z^{\left(r\right)} \right)\in \left\{-1,1\right\}$ of string $z^{\left(r\right)} $, $Y_{n+1}^{\left(r\right)} \in \left\{-1,1\right\}$ is a measured Pauli operator of $r$-th round, while ${\textstyle\frac{d\tilde{\Phi }_{k} }{d{\rm {\mathcal S}}(\vec{\theta }_{r})}} $ is a partial derivative. 

\textbf{Step 5}. (Gate parameter updates). If $r<R$, update gate parameter vector $\vec{\theta }_{r+1} $ via backpropagated side information as 
\begin{equation} \label{ZEqnNum920400} 
\vec{\theta }_{r+1} =\vec{\theta }_{r} -\omega _{r}, 
\end{equation} 
where $\omega _{r}$ is defined as
\begin{equation} \label{ZEqnNum848371} 
\omega _{r} =\frac{\lambda }{r} \sum _{k=1}^{r}g_{k}  ,                                                         
\end{equation} 
where $g_{r} $ is the gradient \eqref{ZEqnNum677383} evaluated in the $k$-th measurement round,
\begin{equation} \label{13)} 
g_{k} =\frac{{\rm {\mathcal L}}(x_{0,k} ,\tilde{l}(z^{\left(k\right)}))}{d{\rm {\mathcal S}}(\vec{\theta }_{k})} ,                                                       
\end{equation} 
while $\lambda $ is the learning rate. 

Set the gate-parameter modification vector $\vec{\alpha }_{r} $ as
\begin{equation} \label{ZEqnNum174287} 
\vec{\alpha }_{r} =\left(\alpha _{r,1} ,\ldots ,\alpha _{r,L} \right)^{T} ,                                                   
\end{equation} 
where $\alpha _{r,i} $ is associated to the modification of the gate parameter $\theta _{r,i} $ of the $i$-th unitary $U_{i} \left(\theta _{r,i} \right)$ as
\begin{equation} \label{15)} 
\alpha _{r,i} =\omega _{r} .                                                            
\end{equation} 

\textbf{Step 6}. (Output gradient). Apply steps 1-5 for all $r$.

Output the $G$ final gradient of the $R$ rounds via the summation of the $R$ gradients as
\begin{equation} \label{ZEqnNum955196} 
G=\sum _{r=1}^{R}g_{r},
\end{equation} 
where $g_{r}$ is given in \eqref{ZEqnNum677383}.
\end{algo}

As a corollary, the training of ${\rm {\mathcal D}}\left({\rm RQNN}_{QG} \right)$ can be reduced to a backpropagation method via the environmental graph of ${\rm RQNN}_{QG} $.

\end{proof}

\subsubsection{Description and Method Validation}
The detailed steps and validation of Algorithm 2 are as follows.

In Step 1, the number $R$ of measurement rounds are set for ${\rm RQNN}_{QG} $. For each measurement round initialization steps \eqref{ZEqnNum398970}-\eqref{ZEqnNum193600} are set.

Step 2 provides the quantum evolution phase of ${\rm RQNN}_{QG} $, and produces output quantum system ${\left| Y_{r}  \right\rangle} $ \eqref{ZEqnNum396843} via forward propagation of quantum information through the unitary sequence $U(\vec{\theta }_{r})$ of the $L$ unitaries.

Step 3 initializes the ${\rm P}^{\left(r\right)} \left({\rm RQNN}_{QG} \right)$ post-processing method via the definition of \eqref{ZEqnNum522731} for gradient computations. In \eqref{ZEqnNum636867}, the quantity $\Phi _{r} =z^{\left(r\right)} +U(\vec{\theta }_{r-1})+B_{r} $ connects the side information of the $r$-th measurement round with the side information of the $\left(r-1\right)$-th measurement round; and $U(\vec{\theta }_{r-1})$ is the unitary sequence of the $\left(r-1\right)$-th round, and $B_{r} $ is a bias the current measurement round. 
The quantity $\xi _{r,k} ={d\Phi _{r}  \mathord{\left/{\vphantom{d\Phi _{r}  d\Phi _{k} }}\right.\kern-\nulldelimiterspace} d\Phi _{k} } $ in \eqref{ZEqnNum809132} utilizes the $\Phi _{i} $ quantities (see \eqref{ZEqnNum636867}) of the $i$-th measurement rounds, such that $i=k+1,\ldots ,r$, where $k<r$.

Step 4 determines the $g_{r} $ loss function gradient of the $r$-th measurement round. In \eqref{ZEqnNum677383}, the $g_{r} $ gradient is determined as $\sum _{k=1}^{r}{\textstyle\frac{{\rm {\mathcal L}}(x_{0,r} ,\tilde{l}(z^{\left(r\right)}))}{d\Phi _{r} }} {\textstyle\frac{d\Phi _{r} }{d\Phi _{k} }} {\textstyle\frac{d\tilde{\Phi }_{k} }{d{\rm {\mathcal S}}(\vec{\theta }_{r})}}  $, that is, via the utilization of the side information of the $k=1,\ldots ,r$ measurement rounds at a particular $r$.

In Step 5, the gate parameters are updated via the gradient descent rule \cite{ref24}, by utilizing the gradients of the $k=1,\ldots ,r$ measurement rounds at a particular $r$. Since in \eqref{ZEqnNum920400} all the gate parameters of the $L$ unitaries are updated by  $\omega _{r} $ as given in \eqref{ZEqnNum848371}, for a particular unitary $U_{i} \left(\theta _{r,i} \right)$, the gate parameter is updated via $\vec{\alpha }_{r} $ \eqref{ZEqnNum174287} to $\theta _{r+1,i}$ as
\begin{equation} \label{17)} 
\theta _{r+1,i} =\theta _{r,i} -\alpha _{r,i} =\theta _{r,i} -\omega _{r} . 
\end{equation}

Finally, Step 6 outputs the $G$ final gradient of the total $R$ measurement rounds in \eqref{ZEqnNum955196}, as a summation of the $g_{r} $ gradients \eqref{ZEqnNum677383} determined in the $r=1,\ldots ,R$ rounds.

The steps of the learning method of an ${\rm RQNN}_{QG} $ (Algorithm 2) are illustrated in \fref{fig3}. The $\vec{\theta }_{r} $ gate parameters of the unitaries of unitary sequence $U(\vec{\theta }_{r})$ are set as $\vec{\theta }_{r} =\vec{\theta }_{r-1} -\omega _{r-1} ,$ where $\vec{\theta }_{r-1} $ is the gate parameter vector associated to sequence $U(\vec{\theta }_{r-1})$, while $\alpha _{r-1,i} =\omega _{r-1} $ is the gate parameter modification coefficient, and $\omega _{r-1} ={\textstyle\frac{\lambda }{r-1}} \sum _{k=1}^{r-1}{\textstyle\frac{{\rm {\mathcal L}}(x_{0,k} ,\tilde{l}(z^{\left(k\right)}))}{d{\rm {\mathcal S}}(\vec{\theta }_{k})}}  $.

 \begin{center}
\begin{figure*}[h!]
%\vspace{-0.4cm}
\begin{center}
\includegraphics[angle = 0,width=1\linewidth]{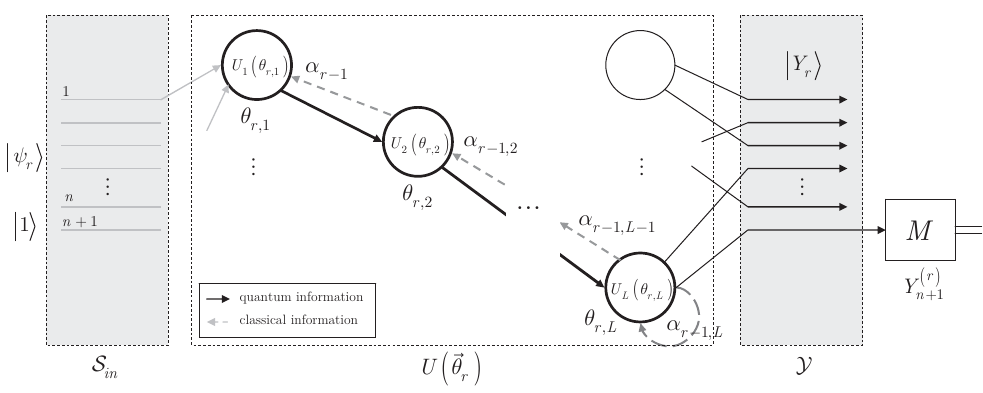}
\caption{The learning method for an ${\rm RQNN}_{QG} $. In an $r$-th measurement round, the ${\rm RQNN}_{QG} $ network realizes the unitary sequence $U(\vec{\theta }_{r})$, and side information is available about the previous $k=1,\ldots ,r-1$ running sequences of the structure.
 Quantum information propagates only forward in the network via quantum links (solid lines), the $\alpha _{r-1,i} =\omega _{r-1} $ quantities are distributed via backpropagation of side information through the classical links (dashed lines). The $\theta _{r,i} $ gate parameter of an $i$-th unitary $U_{i} \left(\theta _{r,i} \right)$ of $U(\vec{\theta }_{r})$ is set to $\theta _{r,i} =\theta _{r-1,i} -\alpha _{r-1,i} ,$ where $\theta _{r-1,i} $ is the gate parameter of the $i$-th unitary $U_{i} \left(\theta _{r-1,i} \right)$ of the $U(\vec{\theta }_{r-1})$ unitary sequence.} 
 \label{fig3}
 \end{center}
\end{figure*}
\end{center}

\subsubsection{Closed-Form Error Evaluation}
\begin{lemma}
The $\delta $ quantity of the unitaries of a ${\rm {\mathcal D}}\left({\rm RQNN}_{QG} \right)$ can be expressed in a closed form via the ${\rm {\mathcal G}}_{{\rm RQNN}_{QG} } $ environmental graph of ${\rm RQNN}_{QG} $. 
\end{lemma}
\begin{proof}
Let ${\rm {\mathcal G}}_{{\rm RQNN}_{QG} } $ be the environmental graph of ${\rm RQNN}_{QG} $, such that ${\rm RQNN}_{QG} $ is characterized via $\vec{\theta }$ (see \eqref{ZEqnNum837426}). Utilizing the structure ${\rm {\mathcal G}}_{{\rm RQNN}_{QG} } $ of ${\rm RQNN}_{QG} $ allows us to express the square error in a closed form as follows.

Let $Y$ and $Z$ refer to output realizations ${\left| Y \right\rangle} $ and ${\left| Z \right\rangle} $ of ${\rm RQNN}_{QG} $, ${\rm {\mathcal Y}}\in Y,Z$, with an output set ${\rm {\mathcal Y}}$, and let ${\rm {\mathcal L}}(x_{0} ,\tilde{l}(z))$ be the loss function. Then let $\mathbf{H}_{{\rm RQNN}_{QG} } $ be a Hessian matrix \cite{ref23} of the ${\rm RQNN}_{QG} $ structure, with a generic coordinate $\hbar _{ij,lm}^{{\rm RQNN}_{QG} } $, as
\begin{equation} \label{98)} 
\begin{split}
  \hbar _{ij,lm}^{\text{RQN}{{\text{N}}_{QG}}}=&\frac{{{d}^{2}}\mathcal{L}( {{x}_{0}},\tilde{l}\left( z \right))}{d{{\theta }_{ij}}d{{\theta }_{lm}}} \\ 
  =& \frac{d}{d{{\theta }_{ij}}}\sum\limits_{Y\in \mathcal{Y}}{\frac{d\mathcal{L}( {{x}_{0}},\tilde{l}\left( z \right))}{d{{W}_{{{U}_{Y}}\left( {{\theta }_{Y}} \right)}}}}\frac{d{{W}_{{{U}_{Y}}\left( {{\theta }_{Y}} \right)}}}{d{{\theta }_{lm}}} \\ 
  =& \sum\limits_{Y\in \mathcal{Y}}{\sum\limits_{Z\in \mathcal{Y}}{\frac{{{d}^{2}}\mathcal{L}( {{x}_{0}},\tilde{l}\left( z \right))}{d{{W}_{{{U}_{Y}}\left( {{\theta }_{Y}} \right)}}d{{W}_{{{U}_{Z}}\left( {{\theta }_{Z}} \right)}}}\delta _{{{U}_{i}}\left( {{\theta }_{i}} \right)}^{Z}\delta _{{{U}_{l}}\left( {{\theta }_{l}} \right)}^{Y}{{W}_{{{U}_{j}}\left( {{\theta }_{j}} \right)}}{{W}_{{{U}_{m}}\left( {{\theta }_{m}} \right)}}}} \\ 
 & +\sum\limits_{Y\in \mathcal{Y}}{\frac{d\mathcal{L}( {{x}_{0}},\tilde{l}\left( z \right))}{d{{W}_{{{U}_{Y}}\left( {{\theta }_{Y}} \right)}}}\left( {{W}_{{{U}_{m}}\left( {{\theta }_{m}} \right)}}\frac{d\delta _{{{U}_{l}}\left( {{\theta }_{l}} \right)}^{Y}}{d{{\theta }_{ij}}}+\delta _{{{U}_{l}}\left( {{\theta }_{l}} \right)}^{Y}\frac{d{{W}_{{{U}_{m}}\left( {{\theta }_{m}} \right)}}}{d{{\theta }_{ij}}} \right)} \\ 
   = & \sum\limits_{Y\in \mathcal{Y}}{\sum\limits_{Z\in \mathcal{Y}}{\frac{{{d}^{2}}\mathcal{L}( {{x}_{0}},\tilde{l}\left( z \right))}{d{{W}_{{{U}_{Y}}\left( {{\theta }_{Y}} \right)}}d{{W}_{{{U}_{Z}}\left( {{\theta }_{Z}} \right)}}}\delta _{{{U}_{i}}\left( {{\theta }_{i}} \right)}^{Z}\delta _{{{U}_{l}}\left( {{\theta }_{l}} \right)}^{Y}{{W}_{{{U}_{j}}\left( {{\theta }_{j}} \right)}}{{W}_{{{U}_{m}}\left( {{\theta }_{m}} \right)}}}} \\ 
  & +  \sum\limits_{Y\in \mathcal{Y}}{\frac{d\mathcal{L}( {{x}_{0}},\tilde{l}\left( z \right))}{d{{W}_{{{U}_{Y}}\left( {{\theta }_{Y}} \right)}}}\left( {{\left( \delta _{{{U}_{l}}\left( {{\theta }_{l}} \right),{{U}_{i}}\left( {{\theta }_{i}} \right)}^{Y} \right)}^{2}}{{W}_{{{U}_{m}}\left( {{\theta }_{m}} \right)}}{{W}_{{{U}_{j}}\left( {{\theta }_{j}} \right)}}+{{f}_{i\angle m}}\left( \delta _{{{U}_{l}}\left( {{\theta }_{l}} \right)}^{Y}\delta _{{{U}_{i}}\left( {{\theta }_{i}} \right)}^{m}{{W}_{{{U}_{j}}\left( {{\theta }_{j}} \right)}} \right) \right)},  
\end{split} 
\end{equation} 
where $W_{U_{i} \left(\theta _{i} \right)}$ is given in \eqref{c}, $f_{i\angle m} \left(\cdot \right)$ is a topological ordering function on ${\rm {\mathcal G}}_{{\rm RQNN}_{QG} } $, indices $Y$ and $Q$ are associated with the output realizations ${\left| Y \right\rangle} $ and ${\left| Q \right\rangle} $, while $\left(\delta _{U_{l} \left(\theta _{l} \right),U_{i} \left(\theta _{i} \right)}^{Q} \right)^{2} $ is the square error between unitaries $U_{l} \left(\theta _{l} \right)$ and $U_{i} \left(\theta _{i} \right)$ at a particular output ${\left| Q \right\rangle} $ as
\begin{equation} \label{ZEqnNum391405} 
\begin{split}
   {{\left( \delta _{{{U}_{l}}\left( {{\theta }_{l}} \right),{{U}_{i}}\left( {{\theta }_{i}} \right)}^{Q} \right)}^{2}}&=\frac{{{d}^{2}}{{W}_{{{U}_{Y}}\left( {{\theta }_{Y}} \right)}}}{d{{Q}_{{{U}_{l}}\left( {{\theta }_{l}} \right)}}d{{Q}_{{{U}_{i}}\left( {{\theta }_{i}} \right)}}}=\frac{d\delta _{{{U}_{i}}\left( {{\theta }_{i}} \right)}^{Y}}{d{{Q}_{{{U}_{l}}\left( {{\theta }_{l}} \right)}}} \\ 
 & =\frac{d\delta _{{{U}_{l}}\left( {{\theta }_{l}} \right)}^{i}}{d{{Q}_{{{U}_{i}}\left( {{\theta }_{i}} \right)}}}\sum\limits_{j\in \Gamma \left( i \right)}{{{\theta }_{ji}}\delta _{{{U}_{l}}\left( {{\theta }_{l}} \right)}^{Y}+{{Q}_{{{U}_{i}}\left( {{\theta }_{i}} \right)}}}\sum\limits_{j\in \Gamma \left( i \right)}{{{\theta }_{ji}}{{\left( \delta _{{{U}_{j}}\left( {{\theta }_{l}} \right),{{U}_{l}}\left( {{\theta }_{l}} \right)}^{Y} \right)}^{2}}},
\end{split}
\end{equation} 
where $Q_{U_{i} \left(\theta _{i} \right)}$ is as in \eqref{qi}.
Note that the relation $\left(\delta _{U_{l} \left(\theta _{l} \right),U_{i} \left(\theta _{i} \right)}^{Q} \right)^{2} \ne 0$ in \eqref{ZEqnNum391405} holds if only there is an edge $s_{il} $ between $v_{U_{i}} \in V$ and $v_{U_{l} \left(\theta _{l} \right)} \in V$ in the environmental graph ${\rm {\mathcal G}}_{{\rm RQNN}_{QG} } $ of ${\rm RQNN}_{QG} $. Thus,
\begin{equation} \label{100)} 
\left(\delta _{U_{l} \left(\theta _{l} \right),U_{i} \left(\theta _{i} \right)}^{Q} \right)^{2} =\left\{\begin{array}{l} {\left(\delta _{U_{l} \left(\theta _{l} \right),U_{i} \left(\theta _{i} \right)}^{Q} \right)^{2} =0,{\rm \; if\; }s_{il} \notin S} \\ {\left(\delta _{U_{l} \left(\theta _{l} \right),U_{i} \left(\theta _{i} \right)}^{Q} \right)^{2} \ne 0,{\rm \; if\; }s_{il} \in S} \end{array}\right. . 
\end{equation} 
Since ${\rm {\mathcal G}}_{{\rm RQNN}_{QG} } $ contains all information for the computation of \eqref{ZEqnNum391405} and ${\rm {\mathcal D}}\left({\rm RQNN}_{QG} \right)$ is defined through the structure of ${\rm {\mathcal G}}_{{\rm RQNN}_{QG} } $, the proof is concluded here.
\end{proof}

\section{Conclusions}
\label{sec6}
Gate-model QNNs allow an experimental implementation on near-term gate-model quantum computer architectures. Here we examined the problem of learning optimization of gate-model QNNs. We defined the constraint-based computational models of these quantum networks and proved the optimal learning methods. We revealed that the computational models are different for nonrecurrent and recurrent gate-model quantum networks. We proved that for nonrecurrent and recurrent gate-model QNNs, the optimal learning is a supervised learning. We showed that for a recurrent gate-model QNN, the learning can be reduced to backpropagation. The results are particularly useful for the training of QNNs on near-term quantum computers.

%
%\section*{Statements}
%\subsection*{Ethics statement}
%This work did not involve any active collection of human data.
%\subsection*{Data accessibility statement}
%This work does not have any experimental data.
%\subsection*{Competing financial interests statement}
%We have no competing financial interests.
%\subsection*{Competing interests statement}
%We have no competing interests.
%\subsection*{Funding}
%No relevant funding. 
%\subsection*{Authors’ contributions}
%L.GY. designed the protocol and wrote the manuscript. L.GY. and S.I. analyzed the results. All authors reviewed the manuscript.
%\subsection*{Original Article Statement}
%a. This paper is a completely novel and completely independent submission, has no any connection with any previously submitted papers.
%b. This manuscript is the authors' original work and has not been published nor has it been submitted simultaneously elsewhere. 
%c. All authors have checked the manuscript and have agreed to the submission. 

\section*{Acknowledgements}
The research reported in this paper has been supported by the National Research, Development and Innovation Fund (TUDFO/51757/2019-ITM, Thematic Excellence Program). This work was partially supported by the National Research Development and Innovation Office of Hungary (Project No. 2017-1.2.1-NKP-2017-00001), by the Hungarian Scientific Research Fund - OTKA K-112125 and in part by the BME Artificial Intelligence FIKP grant of EMMI (BME FIKP-MI/SC).

\newpage
%\onecolumn
\appendix
\setcounter{table}{0}
\setcounter{figure}{0}
\setcounter{equation}{0}
\setcounter{algocf}{0}
\renewcommand{\thetable}{\Alph{section}.\arabic{table}}
\renewcommand{\thefigure}{\Alph{section}.\arabic{figure}}
\renewcommand{\theequation}{\Alph{section}.\arabic{equation}}
\renewcommand{\thealgocf}{\Alph{section}.\arabic{algocf}}

\setlength{\arrayrulewidth}{0.1mm}
\setlength{\tabcolsep}{5pt}
\renewcommand{\arraystretch}{1.5}
\section{Appendix}
\subsection{Abbreviations}
\begin{description}
\item[AI] Artificial Intelligence
\item[DAG] Directed Acyclic Graph
\item[QG] Quantum Gate structure of a gate-model quantum computer
\item[QNN] Quantum Neural Network
\item[RQNN] Recurrent Quantum Neural Network

\end{description}

\subsection{Notations}
The notations of the manuscript are summarized in  \tref{tab2}.
\begin{center}
\begin{longtable}{||l|p{4.5in}||}
\caption{Summary of notations.}
\label{tab2}
\endfirsthead
\endhead
\hline
\textit{Notation} & \textit{Description} \\ \hline
${\rm QNN}_{QG} $ & Quantum neural network implemented on a gate-model quantum computer with a quantum gate structure $QG$. \\ \hline 
${\rm RQNN}_{QG} $ & Recurrent quantum neural network implemented on a gate-model quantum computer with a quantum gate structure $QG$. \\ \hline 
$U_{i} \left(\theta _{i} \right)$ & An $i$-th unitary gate, $U_{i} \left(\theta _{i} \right)=\exp \left(-i\theta _{i} P\right)$, where $P$ is a generalized Pauli operator formulated by a tensor product of Pauli operators $\left\{X,Y,Z\right\}$, while $\theta _{i} $ is referred to as the gate parameter associated to $U_{i} \left(\theta _{i} \right)$. \\ \hline 
$U_{j}(\theta _{ij})$ & Selection of $\theta _{j}$ for the unitary $U_j$ to realize the operation $U_i\left( {{\theta }_{i}} \right)U_j\left( {{\theta }_{j}} \right)$, i.e., the application of $U_j\left( {{\theta }_{j}} \right)$ on the output of $U_i\left( {{\theta }_{i}} \right)$ at a particular gate parameter $\theta _{j}$.\\ \hline 
$U(\vec{\theta })$ & Unitary operator, $U(\vec{\theta })=U_{L} \left(\theta _{L} \right)U_{L-1} \left(\theta _{L-1} \right)\ldots U_{1} \left(\theta _{1} \right)$, where $U_{i} \left(\theta _{i} \right)$ identifies an $i$-th unitary gate. \\ \hline 
$\vec{\theta }$ & A collection of gate parameters of the $L$ unitaries, $\vec{\theta }=\theta _{L} ,\theta _{L-1} ,\ldots ,\theta _{1} $. \\ \hline 
${\left| \psi ,\varphi  \right\rangle} $ & Input system, where ${\left| \psi  \right\rangle} ={\left| z \right\rangle} $ is a computational basis state, where $z$ is an $n$-length string, while the $\left(n+1\right)$-th quantum state initialized as ${\left| \varphi  \right\rangle} ={\left| 1 \right\rangle} $, and is referred to as the readout quantum state.   \\ \hline 
${\left| Y \right\rangle} $ & An $\left(n+1\right)$-length output quantum system of the gate-model quantum neural network. \\ \hline 
$z$ & An $n$-length string, $z=z_{1} z_{2} \ldots z_{n} ,$ where $z_{i} $ represents a classical bit, $z_{i} \in \left\{-1,1\right\}$. \\ \hline 
$f\left(\theta \right)$ & Objective function. \\ \hline 
$l\left(z\right)$ & Binary label of string $z$,  $l\left(z\right)\in \left\{-1,1\right\}$. \\ \hline 
$\tilde{l}$ & Predicted value of the binary label $l\left(z\right)\in \left\{-1,1\right\}$ of string $z$,  $\tilde{l}(z)=\langle  z,1 | {{( U( {\vec{\theta }} ) )}^{\dagger }}{{Y}_{n+1}}U( {\vec{\theta }} )|z,1 \rangle $. \\ \hline 
$\Delta (\tilde{l}(z))$ & Difference of the $\tilde{l}\left(z\right)$ predicted value of the binary label $l\left(z\right)\in \left\{-1,1\right\}$ of the input string $z$, defined as $\Delta (\tilde{l}(z))=|l(z)-\tilde{l}(z)|$, where $\tilde{l}\in [-1,1]$. \\ \hline 
$Y_{n+1} $ & Measured Pauli operator on the ${\left| \varphi  \right\rangle} $ readout quantum state, $Y_{n+1} \in \left\{-1,1\right\}$. \\ \hline 
${\left| Y \right\rangle} ^{\left(r\right)} $ & An output system realization, $r=1,\ldots ,R$, where $R$ is the total number of output instances. \\ \hline 
${\rm {\mathcal S}}_{T} $ & Training set, formulated via $N$ input strings and labels, ${\rm {\mathcal S}}_{T} =\left\{z_{i} ,l\left(z_{i} \right),i=1,\ldots ,N\right\}$. \\ \hline 
${\rm {\mathcal C}}$ & Constraint machine. \\ \hline 
${\rm {\mathcal D}}$ & Diffusion machine. \\ \hline 
${\rm {\mathcal F}}$ & Functional space. \\ \hline 
${\rm {\mathcal G}}$ & Environmental graph, ${\rm {\mathcal G}}=\left(V,S\right)$. A directed acyclic graph (DAG), with a set $V$ of vertexes, and a set $S$ of arcs. \\ \hline 
$V$ & Set of vertexes in the ${\rm {\mathcal G}}$ environmental graph. \\ \hline 
$S$ & Set of arcs in the ${\rm {\mathcal G}}$ environmental graph. \\ \hline 
$v$ & A vertex of  $V$ the ${\rm {\mathcal G}}$ environmental graph. \\ \hline 
$\Gamma \left(v\right)$ & Children set of $v$ the ${\rm {\mathcal G}}$ environmental graph. \\ \hline 
$\left|\Gamma \left(v\right)\right|$ & Cardinality of set $\Gamma \left(v\right)$. \\ \hline 
$\left\langle x\right\rangle $ & An identifier. \\ \hline 
${\rm {\mathcal X}}$ & Perceptual space. \\ \hline 
${\rm {\mathcal Z}}$ & Mapped space. \\ \hline 
$x$ & An element (vector) of the perceptual space ${\rm {\mathcal X}}\subset {\rm {\mathbb{C}}}^{d} $. \\ \hline 
$\diamondsuit $ & Symbol of missing features. \\ \hline 
${\rm {\mathcal X}}_{0} $ & Initial perceptual space, ${\rm {\mathcal X}}_{0} ={\rm {\mathcal X}}\bigcup \left\{\diamondsuit \right\}$. \\ \hline 
${\rm {\mathcal I}}$ & Individual space, ${\rm {\mathcal I}}=V\times {\rm {\mathcal X}}_{0} $. \\ \hline 
$f_{{\rm {\mathcal P}}} $ & A perceptual map, $f_{{\rm {\mathcal P}}} :\tilde{V}\to {\rm {\mathcal X}}:x=f_{{\rm {\mathcal P}}} \left(v\right)$, where $\tilde{V}$ is a subset  $V$ in the ${\rm {\mathcal G}}$ environmental graph. \\ \hline 
$\iota $ & An individual of the individual space ${\rm {\mathcal I}}$, $\iota =\Upsilon x+\neg \Upsilon v$, where $+$ is the sum operator in ${\rm {\mathbb{C}}}^{d} $, while $\Upsilon $ is a constraint as\newline $\Upsilon :(v\in \tilde{V})\vee \left(x\in {\rm {\mathcal X}}\backslash {\rm {\mathcal X}}_{0} \right)$.                                                     \\ \hline 
$C_{\iota ^{*} } $, $C_{\iota } $ & Constraints. \\ \hline 
$\chi \left(\cdot \right)$ & Compact constraint. \\ \hline 
${\rm {\mathcal S}}_{in} $ & Input space. \\ \hline 
${\rm {\mathcal U}}$ & Space of unitaries. \\ \hline 
${\rm {\mathcal Y}}$ & Output space. \\ \hline 
${\rm {\mathcal G}}_{{\rm QNN}_{QG} } $ & Environmental graph of a ${\rm QNN}_{QG} $.  \\ \hline 
${\rm {\mathcal G}}_{{\rm RQNN}_{QG} } $ & Environmental graph of an ${\rm RQNN}_{QG} $. \\ \hline 
$v_{U_{i}} $ & A vertex associated to the unitary $U_{i} \left(\theta _{i} \right)$ in the environmental graph. \\ \hline 
$v_{0} $ & A vertex associated to the input in the environmental graph. \\ \hline 
$\theta _{ij} $ & Gate parameter, associated to the directed arch $s_{ij} $ between $v_{U_{i}} $ and $v_{U_{j}} $. \\ \hline 
$x_{U_{i} \left(\theta _{i} \right)} $ & An element of ${\rm {\mathcal X}}$ associated to unitary $U_{i} \left(\theta _{i} \right)$. \\ \hline 
$x_{0} $ & An element of ${\rm {\mathcal X}}$ associated to the input, $x_{0} ={\left| z,1 \right\rangle} $. \\ \hline 
$a_{U_{i} \left(\theta _{i} \right)} $ & Parameter defined for a $U_{i} \left(\theta _{i} \right)$ as $a_{U_{i} \left(\theta _{i} \right)} =\sum _{h\in \Xi \left(i\right)}U_{h} \left(\theta _{h} \right)x_{U_{h} \left(\theta _{h} \right)} +b_{U_{i} \left(\theta _{i} \right)}  ,$ where $\Xi \left(i\right)$ refers to the parent set of $v_{U_{i}} $, $b_{U_{i} \left(\theta _{i} \right)} $ is the bias relative to $v_{U_{i}} $. \\ \hline 
$f_{\angle } \left(\cdot \right)$ & Topological ordering function on the environmental graph. \\ \hline 
$\mathbf{H}$ & Hessian matrix. \\ \hline 
$f_{T} \left(\cdot \right)$ & Transition function. \\ \hline 
$F_{O} \left(\cdot \right)$ & Output function. \\ \hline 
$\gamma $ & State variable in the mapped space ${\rm {\mathcal Z}}$, $\gamma \in {\rm {\mathcal Z}}$. \\ \hline 
$\gamma _{v} $ & State variable associated to $v$, $\gamma _{v} \in {\rm {\mathcal Z}}$. \\ \hline 
$\phi \left(A\right)$ & Associated function-pair, $\phi \left(A\right)=\left(f_{T} ,F_{O} \right)$. \\ \hline 
${\left| \gamma _{v}  \right\rangle} $ & System state associated to state variable $\gamma _{v} $. \\ \hline 
$\zeta _{v} $ & Constraint on $f_{T} \left({\rm QNN}_{QG} \right)$ for a ${\rm QNN}_{QG} $. \\ \hline 
$\wp _{v} $ & Constraint on $F_{O} \left({\rm QNN}_{QG} \right)$ for a ${\rm QNN}_{QG} $. \\ \hline 
$\circ $ & Composition operator, $\left(f\circ g\right)\left(x\right)=f\left(g\left(x\right)\right)$. \\ \hline 
$\alpha $ & Parameter. \\ \hline 
$\pi _{v} $ & Compact constraint on $f_{T} \left({\rm QNN}_{QG} \right)$ and $F_{O} \left({\rm QNN}_{QG} \right)$. \\ \hline 
$\Lambda _{v} $ & Constraint on $f_{T} \left({\rm RQNN}_{QG} \right)$ of ${\rm RQNN}_{QG} $. \\ \hline 
$\Omega _{v} $ & Constraint on $F_{O} \left({\rm RQNN}_{QG} \right)$ of ${\rm RQNN}_{QG} $. \\ \hline 
${\mathchar'26\mkern-10mu\lambda} \left(f\left(x\right)\right)$ & Diffuse constraint for ${\rm RQNN}_{QG} $. \\ \hline 
$H_{t} $ & Unit vector for a unitary $U_{t} \left(\theta _{t} \right)$, $t=1,\ldots ,L-1$, $H_{t} =x_{t} +iy_{t} $, where $x_{t} $ and $y_{t} $ are real values. \\ \hline 
$Z_{t+1} $ & System state, $Z_{t+1} =U(\vec{\theta })H_{t} +Ex_{t+1} $, where $E$ is a basis vector matrix. \\ \hline 
$f_{\sigma }^{{\rm RQNN}_{QG} } \left(\cdot \right)$ & Function for ${\rm RQNN}_{QG} $. \\ \hline 
$\left|\right|_{1} $ & $L1$-norm. \\ \hline 
$W$ & An output matrix. \\ \hline 
$D$ & Jacobian matrix. \\ \hline 
$A$ & Constraint matrix. \\ \hline 
$b\left(x\right)$ & Smooth vector-valued function with compact support. \\ \hline 
$f^{*} $ & Compact function subject to be determined. \\ \hline 
${{{\mathcal{S}}_{L\left( \text{QNN} \right)}}}$ & Non-empty supervised learning set. \\ \hline 
$\ell $ & Differential operator, $\ell ={{P}^{\dagger }}P$, where ${{P}^{\dagger }}$ is the adjoint of $P$. \\ \hline 
${{\nabla }^{2}}$ & Laplacian operator. \\ \hline 
${\rm {\mathcal G}}\left(\cdot \right)$ & Green function. \\ \hline 
 $\mathcal{L}$ & Lagrangian. \\ \hline
$\lambda \left(x\right)$ & Lagrange multiplier. \\ \hline 
$H\left(x\right)$, $\Phi $, $\chi _{\kappa } $ & Parameters used in the calculation of compact function $f^{*} \left(x\right)$. \\ \hline 
${\rm {\mathcal L}}(x_{0} ,\tilde{l}(z))$ & Loss function.  \\ \hline 
${\rm {\mathcal A}}_{{\rm {\mathcal C}}\left({\rm QNN}_{QG} \right)} $ & Learning method for ${\rm {\mathcal C}}\left({\rm QNN}_{QG} \right)$. \\ \hline 
${\rm {\mathcal A}}_{{\rm {\mathcal D}}\left({\rm RQNN}_{QG} \right)} $ & Learning method for ${\rm {\mathcal D}}\left({\rm RQNN}_{QG} \right)$. \\ \hline 
$\vec{T}_{{\rm {\mathcal G}}_{{\rm QNN}_{QG} } } $ & Topologically sorted node set, $\vec{T}_{{\rm {\mathcal G}}_{{\rm QNN}_{QG} } } =\left(q_{1} ,\ldots ,q_{L} \right)$. \\ \hline 
$\Xi \left(k\right)$ & Parents of $k\in V$ in the environmental graph. \\ \hline 
${\rm P}^{\left(r\right)} ({\rm {\mathcal G}}_{{\rm QNN}_{QG} })$ & Post-processing associated to the $r$-th measurement on ${\rm {\mathcal G}}_{{\rm QNN}_{QG} } $. \\ \hline 
$\delta _{U_{i} \left(\theta _{i} \right)} $ & Error associated to unitary $U_{i} \left(\theta _{i} \right)$ in the environmental graph. \\ \hline 
$\nu _{U_{L} \left(\theta _{L} \right)} $ & Vertex associated to $U_{L} \left(\theta _{L} \right)$ in the environmental graph. \\ \hline 
$\vec{\Delta }\theta _{i} $ & Parameter modification. \\ \hline 
$g_{U_{i} \left(\theta _{i} \right),U_{j} \left(\theta _{j} \right)} $ & Gradient between unitaries. \\ \hline 
$S_{i} $ & Structure from the environmental graph. \\ \hline 
$\vec{\delta }_{S_{i} } $ & Error vector associated to structure $S_{i} $. \\ \hline 
$\eta $ & Learning parameter. \\ \hline 
$S'$ & Children structure. \\ \hline 
$I$ & Identity operation. \\ \hline 
${\left| Q \right\rangle} $ & An output realization of ${\rm RQNN}_{QG} $. \\ \hline 
$\mathbf{H}_{{\rm RQNN}_{QG} } $ & Hessian matrix of the ${\rm RQNN}_{QG} $ structure. \\ \hline 
$\hbar _{ij,lm}^{{\rm RQNN}_{QG} } $ & A generic coordinate of the Hessian matrix $\mathbf{H}_{{\rm RQNN}_{QG} } $. \\ \hline 
$(\delta _{U_{l} (\theta _{l} ),U_{i} (\theta _{i} )}^{Q} )^{2} $ & Square error between unitaries  $U_{l} \left(\theta _{l} \right)$ and $U_{i} \left(\theta _{i} \right)$ at a particular output ${| Q \rangle} $ of ${\rm RQNN}_{QG} $. \\ \hline
 \end{longtable}
\end{center}
\end{document}